\documentclass[journal,draftclsnofoot,onecolumn,12pt,twoside]{IEEEtranTCOM}
\normalsize

\usepackage{ifpdf}
\usepackage{cite}
\ifpdf
\usepackage[pdftex]{graphicx}
	\graphicspath{{./FigurePDF/}}
	\DeclareGraphicsExtensions{.pdf}
\else
  \usepackage[dvips]{graphicx}
  \graphicspath{{./FigureEPS/}}
  \DeclareGraphicsExtensions{.eps}
\fi
\usepackage[cmex10]{amsmath}
\usepackage{amssymb}
\usepackage{wasysym}
\usepackage{mathrsfs}
\usepackage{stmaryrd}
\usepackage{upgreek}
\usepackage{enumerate}
\DeclareMathAlphabet{\mathpzc}{OT1}{pzc}{m}{it}
\usepackage{algorithmic}
\usepackage{array}
\usepackage[tight,footnotesize]{subfigure}
\usepackage{stfloats}
\usepackage{url}
\usepackage{multirow}
\usepackage{algorithm}
\usepackage{algorithmic}
\usepackage{color}
\usepackage{bbold}

\newtheorem{definition}{Definition}
\newtheorem{te}{Theorem}

\newtheorem{prop}{Proposition}

\newcommand{\mvtoc}{\mu}
\newcommand{\mctov}{\eta}
\newcommand{\pvtoc}{q}
\newcommand{\pctov}{r}
\newcommand{\Phiv}{\Phi^{(v)}}
\newcommand{\Phic}{\Phi^{(c)}}
\newcommand{\Phia}{\Phi^{(a)}}

\newcommand{\fig}[4]{ \begin{figure}[#4]
  \centering
   \includegraphics[width=#3\textwidth]{figures/#1}
   \caption{#2}\label{fig:#1}
  \end{figure}
}

\begin{document}
\title{Analysis and Design of Finite Alphabet Iterative Decoders Robust to Faulty Hardware}
\author{Elsa Dupraz, David Declercq, Bane Vasi\'{c} and Valentin Savin
\thanks{This work was funded by the Seventh Framework Programme of the European Union, under Grant Agreement number 309129 (i-RISC project), and by the NSF under grants CCF-0963726 and CCF-1314147.}
\thanks{E. Dupraz and D. Declercq are with the ETIS lab, ENSEA/Universit\'{e} de Cergy-Pontoise/CNRS UMR 8051, 95014 Cergy-Pontoise, France (e-mail:elsa.dupraz@ensea.fr; declercq@ensea.fr).}
\thanks{B. Vasi$\acute{\mathrm{c}}$ is with the Department
of Electrical and Computer Engineering, University of Arizona, Tucson, AZ, 85721 USA (e-mail: vasic@ece.arizona.edu).}
\thanks{V. Savin is with the CEA-LETI, Minatec Campus, 38000 Grenoble, France (e-mail:valentin.savin@cea.fr)}
\markboth{Submitted to IEEE Transactions on Communications, October 2014.}{Dupraz \MakeLowercase{\textit{et al.}}: Analysis and Design of Finite Alphabet Iterative Decoders Robust to Faulty Hardware}}

\maketitle

\begin{abstract}
This paper addresses the problem of designing LDPC decoders robust to transient errors introduced by a faulty hardware.
We assume that the faulty hardware introduces errors during the message passing updates and we propose a general framework for the definition of the message update faulty functions.
Within this framework, we define symmetry conditions for the faulty functions, and derive two simple error models used in the analysis.
With this analysis, we propose a new interpretation of the functional Density Evolution threshold introduced in~\cite{Kameni14Coms,ngassa14ITA}, and show its limitations in case of highly unreliable hardware. 
However, we show that under restricted decoder noise conditions, the functional threshold can be used to predict the convergence behavior of FAIDs under faulty hardware.
In particular, we reveal the existence of robust and non-robust FAIDs and propose a framework for the design of robust decoders.
We finally illustrate robust and non-robust decoders behaviors of finite length codes using Monte Carlo simulations.
\end{abstract}

\section{Introduction}\label{sec:intro}
Reliability is becoming a major issue in the design of modern electronic devices.
The huge increase in integration factors coupled with the important reduction of the chip sizes makes the devices much more sensitive to noise and may induce transient errors.
Furthermore, the fabrication process makes hardware components more prone to defects and may also cause permanent computation errors.
As a consequence, in the context of communication and storage, errors may not only come from transmission channels, but also from the faulty hardware used in transmitters and receivers.

The general problem of reliable function computation using faulty gates was first addressed by von Neumann in~\cite{von1956probabilistic} and the notion of redundancy was later considered in~\cite{gacs1994lower,dobrushin1977upper,pippenger1985networks}.
Hardware redundancy is defined as the number of noisy gates required for reliable function computation divided by the number of noiseless gates needed for the same function computation.
G\'{a}cs and G\'{a}l~\cite{gacs1994lower} and Dobrushin and Ortyukov~\cite{dobrushin1977upper}, respectively, provided lower and upper bounds on the hardware redundancy for reliable Boolean function computation from faulty gates.
Pippenger~\cite{pippenger1985networks} showed that finite asymptotic redundancy can be achieved when using Low Density Parity Check (LDPC) codes for the reliable computation of linear Boolean functions.
Taylor~\cite{taylor1968reliablestorage} and Kuznetsov~\cite{kuznetsov1973information} considered memories as a particular instance of this problem
and provided an analysis of a memory architecture based on LDPC decoders made of faulty components.
More recently, an equivalence between the architecture proposed by Taylor and a noisy Gallager-B decoder was identified by Vasic \emph{et al.}~\cite{vasic2007information}, while Chilappagari \emph{et al.}~\cite{chilappagari2006analysis} analyzed a memory architecture based on one-step majority logic decoders.

As a consequence, there is a need to address the problem of constructing reliable LDPC decoders made of faulty components not only for error correction on faulty hardware, but also as a first step in the context of reliable function computation and storage.
Formulating a general method for construction of robust decoders requires understanding whether a particular decoder is inherently robust to errors introduced by the faulty hardware.
There is also a need for a rigorous analysis to determine which characteristics of decoders make them robust.

To answer to the first point, Varshney~\cite{varshney2011performance} introduced a framework referred to as noisy Density Evolution (noisy-DE) for the performance analysis of noisy LDPC decoders in terms of asymptotic error probability.
Based on this framework, the asymptotic performance of a variety of noisy LDPC decoders was analyzed.
In~\cite{varshney2011performance}, infinite precision BP decoders were investigated, which is not useful for actual implementation on faulty hardware.
On the contrary, noisy practically important hard-decision decoders, such as noisy Gallager-A~\cite{varshney2011performance} and Gallager-E~\cite{huang2013analysis} decoders were considered.
Gallager-B decoders were analyzed for binary~\cite{vasic2007information,leduc2012faulty,huang2013gallager} and non-binary~\cite{yazdi2013gallager} alphabets under transient error models, and ~\cite{huang2013gallager} also considered permanent error models.
From the same noisy-DE framework,~\cite{Balatsoukas14ComL,ngassa2013min} proposed an asymptotic analysis of the behavior of stronger discrete Min-Sum decoders, for which the exchanged messages are no longer binary but are quantized soft information represented by a finite (and typically small) number of bits.

Recently, a new class of LDPC decoders referred to as Finite Alphabet Iterative Decoders (FAIDs) has been introduced~\cite{Planjery_IEEETransCommun_2013}. 
In these decoders, the messages take their values in small alphabets and the variable node update is derived through a predefined Boolean function.
The FAID framework offers the possibility to define a large collection of these functions, each corresponding to a particular decoding algorithm.
The FAIDs were originally introduced to address the error floor problem, and designed to correct error events located on specific small topologies of error events referred to as {\it trapping sets} that usual decoders (Min-Sum, BP-based) cannot correct.
When operating on faulty hardware, the FAIDs may potentially exhibit very different properties in terms of tolerance to transient errors and we are interested in identifying the robust ones among the large diversity of decoders.

In this paper, we propose a rigorous method for the analysis and the design of decoding rules robust to transient errors introduced by the hardware.
We assume that the faulty hardware introduces transient errors during function computation and propose a general description of faulty functions.
We introduce new symmetry conditions for faulty functions that are more general than those in~\cite{varshney2011performance}. 
We discuss possible simplifications of the general description and present two particular error models to represent the faulty hardware effect.
The design procedure we propose is based on an asymptotic performance analysis of noisy-FAIDs using noisy-DE.
In order to characterize the asymptotic behavior of the FAIDs from the noisy-DE equations, we follow the definition of the noisy-DE threshold of~\cite{Kameni14Coms,ngassa14ITA}, referred to as the \emph{functional threshold}.
We analyze more precisely the behavior of the functional threshold and we observe that if the decoder noise level is too high, the functional threshold fails at predicting the convergence behavior of the faulty decoder. 
However, under the restricted decoder noise conditions, we show that the functional threshold can be used to predict the behavior of noisy-FAIDs and gives a criterion for the comparison of the asymptotic performance of the decoders.
Based on this criterion, we then propose a noisy-DE based framework for the design of decoders inherently robust to errors introduced by the hardware.
Finite-length simulations illustrate the gain in performance at considering robust FAIDs on faulty hardware.

The outline of the paper is as follows.
Section~\ref{sec:FaultyFAID} gives the notations and basic decoder definition.
Section~\ref{sec:models} introduces a general description of faulty functions and presents particular error models.
Section~\ref{sec:FaultyDE} gives the noisy-DE analysis for particular decoder noise models.
Section~\ref{sec:th} restates the definition of the functional threshold and presents the analysis of its behavior.
Section~\ref{sec:selection} presents the method for the design of robust decoders.
Section~\ref{sec:results} gives the finite-length simulation results, and Section~\ref{sec:conclusion} provides the conclusions.

\section{Notations and Decoders Definition}\label{sec:FaultyFAID}
This section introduces notations and basic definitions of FAIDs introduced in~\cite{Planjery_IEEETransCommun_2013}.
In the following, we assume that the transmission channel is a Binary Symmetric Channel (BSC) with parameter $\alpha$.
We consider a BSC because on the hardware all the operations are performed at a binary level.

An $N_s$-level FAID is defined as a 5-tuple given by $\mathrm{D}=(\mathcal{M},\mathcal{Y},\Phiv,\Phic,\Phia)$.
The message alphabet is finite and can be defined as $\mathcal{M} = \{-L_s,\ldots,-L_1,0,L_1,\ldots,L_s\}$, where $L_i\in\mathbb{R^{+}}$
and $L_i>L_j$ for any $i>j$. It thus consists of $N_s=2s+1$ levels to which the message values belong.
For the BSC, the set $\mathcal{Y}$, which denotes the set of possible channel values, is defined as $\mathcal{Y}=\{\pm \mathrm{B}\}$.
The channel value $y\in \cal{Y}$ corresponding to Variable Node (VN) $v$ is determined based on its
received value. Here, we use the mapping $0\rightarrow +\mathrm{B}$ and $1\rightarrow-\mathrm{B}$.
In the following, $\mvtoc_1, \dots, \mvtoc_{d_c-1}$ denote the values of incoming messages to a Check Node (CN) of degree $d_c$ and let $\mctov_1, \dots, \mctov_{d_v-1}$ be the values of incoming messages to a VN of degree $d_v$.
Denote $\boldsymbol{\mvtoc} = [\mu_1,\dots, \mu_{d_c-1}]$ and $\boldsymbol{\mctov} = [\mctov_1,\dots,\mctov_{d_v-1}]$ the vector representations of the incoming messages to a CN and to a VN, respectively.
FAIDs are iterative decoders and as a consequence, messages $\boldsymbol{\mvtoc}$ and $\boldsymbol{\mctov}$ are computed at each iteration.
However, for simplicity, the current iteration is not specified in the notations of the messages.

At each iteration of the iterative decoding process, the following operations are performed on the messages.
The Check Node Update (CNU) function $\Phic: \mathcal{M}^{d_c-1} \to \mathcal{M} $ is used for the message update at a CN of degree $d_c$.
The corresponding outgoing message is computed as
\begin{equation}\label{eq:Phi_c}
 \mctov_{d_c} = \Phic(\boldsymbol{\mvtoc}) .
\end{equation}
In~\cite{Planjery_IEEETransCommun_2013}, $\Phic$ corresponds to the CNU of the standard Min-Sum decoder.
The Variable Node Update (VNU) function $\Phiv: \mathcal{M}^{d_v-1} \times \mathcal{Y} \to \mathcal{M}$
is used for the update at a VN of degree $d_v$.
The corresponding outgoing message is computed as 
\begin{equation}\label{eq:Phi_v}
 \mvtoc_{d_v} = \Phiv(\boldsymbol{\mctov},y) .
\end{equation}
The properties that $\Phi_v$ must satisfy are given in~\cite{Planjery_IEEETransCommun_2013}. 
At the end of each decoding iteration, the \emph{A Posteriori} Probability (APP) computation produces messages $\gamma$ calculated from the function $\Phia :\mathcal{M}^{d_v} \times \mathcal{Y} \to \mathcal{\bar{M}}$, where $\bar{\mathcal{M}} = \{ -L_{s'}, \dots, L_{s'}\}$ is a discrete alphabet of $N_{s'} = 2s'+1$ levels.
Denote $\boldsymbol{\mctov^{\star}} = [\mctov_1,\dots,\mctov_{d_v}]$ the vector representation of all the messages incoming to a VN. 
The APP computation produces
\begin{equation}\label{eq:Phi_app}
 \gamma = \Phia(\boldsymbol{\mctov^{\star}},y) .
\end{equation}
The APP is usually computed on a larger alphabet $\mathcal{\bar{M}}$ in order to limit the impact of saturation effects when calculating the APP.
The mapping $\Phia$ is given by
\begin{equation}
 \Phia(\tilde{\boldsymbol{\mctov}}^{\star},y)  =  \sum \boldsymbol{\tilde{\boldsymbol{\mctov}}^{\star}}  + y ~~ .
\end{equation}
The hard-decision bit corresponding to each variable node $v_n$ is given by the sign of the APP.
If $\Phia(\tilde{\boldsymbol{\mctov}}^{\star},y) = 0$, then the hard-decision bit is selected at random and takes value $0$ with probability $1/2$.

Alternatively, $\Phiv$ can be represented as a Look-Up Table (LUT). 
For instance, Table \ref{tab:mapping} shows an example of LUT for a 7-level FAID and column-weight three codes when the channel value is $-\mathrm{B}$. 
The corresponding LUT for the value $+\mathrm{B}$ can be deduced by symmetry.
Classical decoders such as the standard Min-Sum and the offset Min-Sum can also be seen as instances of FAIDs.
It indeed suffices to derive the specific LUT from the VNU functions of these decoders. 
Table~\ref{tab:3biMinSum} gives the VNU of the $7$-level offset Min-Sum decoder.
Therefore, the VNU formulation enables to define a large collection of decoders with common characteristics but potentially different robustness to noise.
In the following, after introducing error models for the faulty hardware, we describe a method for analyzing the asymptotic behavior of noisy-FAIDs. 
This method enables us to compare decoder robustness for different mappings $\Phiv$ and thus to design decoders robust to faulty hardware.

\begin{table*}[t]
\parbox{.48\linewidth}{%
\caption{LUT $\Phiv_{ \text{opt}}$ reported in \cite{Planjery_IEEETransCommun_2013} optimized for the error floor}
\centering 
\resizebox{8cm}{!}{
\begin{tabular}{|c||c|c|c|c|c|c|c|}
\hline
\boldmath $m_{1}/m_{2}$    & \boldmath$-L_3$ &\boldmath$-L_2$ & \boldmath$-L_1$  & \boldmath $0$  & \boldmath$+L_1$  & \boldmath $+L_2$ & \boldmath $+L_3$\\ \hline\hline

\boldmath $-L_3$ & $-L_3$ & $-L_3$ & $-L_3$ & $-L_3$ & $-L_3$ & $-L_3$ & $-L_1$ \\  \hline
\boldmath $-L_2$ & $-L_3$ & $-L_3$ & $-L_3$ & $-L_3$ & $-L_2$ & $-L_1$ & $L_1$ \\  \hline
\boldmath $-L_1$ & $-L_3$ & $-L_3$ & $-L_2$ & $-L_2$ & $-L_1$ & $-L_1$ & $L_1$ \\  \hline
\boldmath 0 & $-L_3$ & $-L_3$ & $-L_2$ & $-L_1$ & $0$ & $0$ & $L_1$ \\  \hline
\boldmath $L_1$ & $-L_3$ & $-L_2$ & $-L_1$ & $0$ & $0$ & $L_1$ & $L_2$ \\  \hline
\boldmath $L_2$ & $-L_3$ & $-L_1$ & $-L_1$ & $0$ & $L_1$ & $L_1$ & $L_3$ \\  \hline
\boldmath $L_3$ & $-L_1$ & $L_1$ & $L_1$ & $L_1$ & $L_2$ & $L_3$ & $L_3$ \\ \hline
 \end{tabular}
}
\label{tab:mapping}}%
\hfill
\parbox{.48\linewidth}{%
\caption{VNU of a 3-bit offset Min-Sum represented as a FAID}
\centering \resizebox{8cm}{!}{
\begin{tabular}{|c||c|c|c|c|c|c|c|}
\hline
\boldmath $m_{1}/m_{2}$ & \boldmath$-L_3$ &\boldmath$-L_2$ & \boldmath$-L_1$  & \boldmath $0$  & \boldmath$+L_1$  & \boldmath $+L_2$ & \boldmath $+L_3$\\ \hline\hline
\boldmath $-L_3$ & $-L_3$ & $-L_3$ & $-L_3$ & $-L_3$ & $-L_3$ & $-L_2$ & $-L_1$ \\  \hline
\boldmath $-L_2$ & $-L_3$ & $-L_3$ & $-L_3$ & $-L_3$ & $-L_2$ & $-L_1$ & 0 \\  \hline
\boldmath $-L_1$ & $-L_3$ & $-L_3$ & $-L_3$ & $-L_2$ & $-L_1$ & 0 & 0 \\  \hline
\boldmath 0 & $-L_3$ & $-L_3$ & $-L_2$ & $-L_1$ & 0 & 0 & 0 \\  \hline
\boldmath $L_1$ & $-L_3$ & $-L_2$ & $-L_1$ & 0 & 0 & 0 & $L_1$ \\  \hline
\boldmath $L_2$ & $-L_2$ & $-L_1$ & 0 & 0 & 0 & $L_1$ & $L_2$ \\  \hline
\boldmath $L_3$ & $-L_1$ & 0 & 0 & 0 & $L_1$ & $L_2$ & $L_3$ \\
\hline
 \end{tabular}
}
\label{tab:3biMinSum}}%
\end{table*}

\section{Error Models for Faulty Hardware}\label{sec:models}
In this paper, we assume that the faulty hardware introduces transient errors only during function computation.
For the performance analysis of faulty decoders, specific error models have been considered in previous works.
In~\cite{huang2013analysis,huang2013gallager,Balatsoukas14ComL}, transient errors are assumed to appear at a binary level on message wires between VNs and CNs.
In~\cite{varshney2011performance,ngassa2013min}, the noise effect is represented by a random variable independent of the function inputs and applies only through a deterministic error injection function.
Here we propose a more general error model which includes the above cases.

For the noisy-DE analysis, the considered faulty functions have to be symmetric, which implies that the error probability of the decoder does not change when flipping a codeword symbol.
As a consequence, the error probability of the decoder does not depend on the transmitted codeword, which greatly simplifies the analysis. 
Here, we introduce new symmetry conditions for the general error models.
We then discuss possible simplifications of the general model and introduce two particular simple error models which allow the asymptotic analysis of the faulty iterative decoding.

\subsection{General Faulty Functions and Symmetry Conditions}\label{sec:genmodel}
To describe general faulty functions, we replace the deterministic functions $\Phi^{(c)}$, $\Phiv$, $\Phi^{(a)}$ introduced in Section~\ref{sec:FaultyFAID} by the following conditional Probability Mass Functions (PMF). 
Denote $\tilde{\mvtoc}_{d_v}$, $\tilde{\mctov}_{d_c}$, and $\tilde{\gamma}$ the noisy versions of $\mvtoc_{d_v}$, $\mctov_{d_c}$, $\gamma$, and denote $\tilde{\boldsymbol{\mvtoc}} = [\tilde{\mvtoc}_1,\dots, \tilde{\mvtoc}_{d_c-1}] $, $\tilde{\boldsymbol{\mctov}} = [\tilde{\mctov}_1,\dots, \tilde{\mctov}_{d_v-1}] $, $\tilde{\boldsymbol{\mctov}}^{\star} = [\tilde{\mctov}_1,\dots, \tilde{\mctov}_{d_v}] $ their vector representations.
Then a faulty VNU is defined as the conditional PMF
\begin{equation}\label{eq:sVNU}
\text{P}^{(v)}(\tilde{\mvtoc}_{d_v}|\tilde{\boldsymbol{\mctov}},y ), 
\end{equation}
a faulty CNU is defined as 
\begin{equation}\label{eq:sCNU}
\text{P}^{(c)}(\tilde{\mctov}_{d_c}| \tilde{\boldsymbol{\mvtoc}}) , 
\end{equation}
and a faulty APP is defined as 
\begin{equation}\label{eq:sAPP}
\text{P}^{(a)}(\tilde{\gamma}|\tilde{\boldsymbol{\mctov}}^{\star},y ) . 
\end{equation}
The described model is memoryless and takes only into account transient errors in the decoder, but it ignores permanent errors and possible dependencies with previous or future function arguments.
However it is general enough to represent any type of memoryless mapping and error model. 

For the noisy-DE analysis, the considered faulty functions have to be symmetric.
The definitions of symmetry given in~\cite{varshney2011performance} only consider the particular case of error injection functions and are not sufficient to characterize the symmetry of the above faulty functions.
In the following, we introduce more general definitions of symmetry.

\begin{definition}\label{def:sym}
\begin{enumerate}
 \item A faulty VNU is said to be symmetric if
 \begin{equation}\label{eq:randvnusym}
  \text{P}^{(v)}(\tilde{\mvtoc}_{d_v}|\tilde{\boldsymbol{\mctov}},y ) = \text{P}^{(v)}(-\tilde{\mvtoc}_{d_v}|-\tilde{\boldsymbol{\mctov}},-y ) .
 \end{equation}
 \item A faulty CNU is said to be symmetric if
  \begin{equation}\label{eq:randcnusym}
  \text{P}^{(c)}(\tilde{\mctov}_{d_c}|\mathbf{a}.\tilde{\boldsymbol{\mvtoc}}) = \text{P}^{(c)}\left( \left( \prod \mathbf{a} \right) \tilde{\mctov}_{d_c}| \tilde{\boldsymbol{\mvtoc}}\right) .
 \end{equation}
where $\mathbf{a} = [a_1,\dots,a_{d_c-1}]$, $a_i \in \{-1,1\}$, $\mathbf{a}.\tilde{\boldsymbol{\mvtoc}}$ is the component by component product of $\mathbf{a} $ and $\tilde{\boldsymbol{\mvtoc}} $, and $\prod \mathbf{a}$ is the product of all components in vector $\mathbf{a}$.
 \item A faulty APP is said to be symmetric if 
  \begin{equation}\label{eq:randappsym}
 \text{P}^{(a)}(\tilde{\mvtoc}_{d_v}|\tilde{\boldsymbol{\mctov}}^{\star},y ) = \text{P}^{(a)}(-\tilde{\mvtoc}_{d_v}|-\tilde{\boldsymbol{\mctov}}^{\star},-y ) .
 \end{equation}
\end{enumerate}
\end{definition}
Note that our definitions of symmetry are the same as the ones originally introduced in~\cite{richardson01IT2} for deterministic decoders, except that ours apply on conditional PMFs instead of deterministic mappings.

\subsection{Faulty Function Decomposition}
\fig{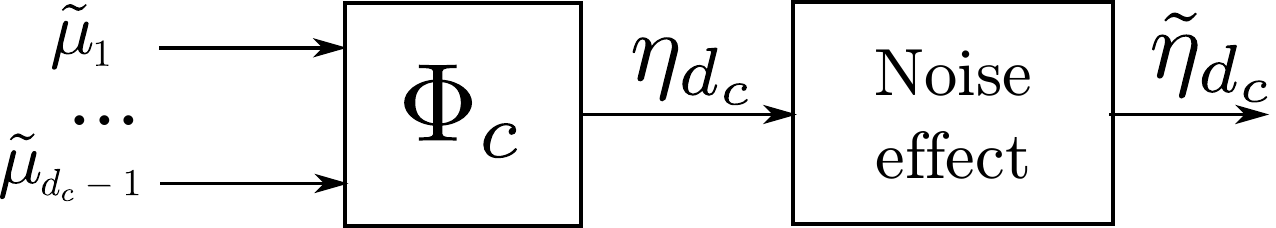}{Function decomposition for the CNU}{0.4}{t}
A possible simplification of the general models described in the previous section is to consider that the noise appears only at the output of a function computation.
More precisely, we assume that the noisy function can be decomposed as a noiseless function followed by the noise effect, as in Fig.~\ref{fig:noise_effect} for the case of the CNU.
In this simplified error model,  $\mctov_{d_c}$, $\mvtoc_{d_v}$, and $\gamma$, represent the messages at the output of the noiseless CNU, VNU, and APP computation respectively, and their noisy versions are denoted $\tilde{\mctov}_{d_c}$, $\tilde{\mvtoc}_{d_v}$, $\tilde{\gamma}$.
The noisy output is assumed to be independent of the inputs conditionally to the noiseless output,~\emph{i.e.}, for the case of faulty CNU, this gives
 $\text{P}^{(c)}(\tilde{\mctov}_{d_c}|\mctov_{d_c},\tilde{\boldsymbol{\mvtoc}} ) = \text{P}^{(c)}(\tilde{\mctov}_{d_c}|\mctov_{d_c} )$.
Furthermore, as the noiseless output is obtained from a deterministic function of the inputs, we get
\begin{equation}\label{eq:condind}
 \text{P}^{(c)}(\tilde{\mctov}_{d_c}|\tilde{\boldsymbol{\mvtoc}} ) = \text{P}^{(c)}(\tilde{\mctov}_{d_c}|\Phic(\tilde{\boldsymbol{\mvtoc}}) ) .
\end{equation}
The same conditions hold for the faulty VNU and APP.

The noise effects at the output of $\Phic$ and $\Phiv$ are represented by probability transition matrices $\Pi^{(v)}$ and $\Pi^{(c)}$ respectively, with
\begin{equation}\label{eq:defptm}
 \Pi^{(c)}_{k,m} = \text{Pr}( \tilde{\mctov}_{d_c} = m | \mctov_{d_c} = k) , ~ \Pi_{k,m}^{(v)} = \text{Pr}( \tilde{\mvtoc}_{d_v} = m | \mvtoc_{d_v} = k ), ~~~ \forall k,m \in \mathcal{M}
\end{equation}
wherein the matrix entries are indexed by the values in $\mathcal{M}$.
This indexing is used for all the vectors and matrices introduced in the remaining of the paper.
The noise effect on $\Phia$ is modeled by the probability transition matrix $\Pi^{(a)}$ with
\begin{equation}
 \Pi^{(a)}_{k, m} = \text{Pr}( \tilde{\gamma} = m | \gamma = k), ~~~  \forall k,m \in \bar{\mathcal{M}} .
\end{equation}
The forms of the probability transition matrices depend on the considered error models.
In the next section, two simple examples derived from this simplified model are introduced.
They will then be considered in the noisy-DE analysis.

Note that in the above decomposition model the noise is added only at a message level at the output of the noiseless functions.
An alternative model would be to consider noise effect introduced \emph{inside} the functions, for example during elementary operations such as the minimum computation between two elements in $\Phic$, as in~\cite{ngassa2013min}.
While the decomposition model introduced here may not capture all the noise effects, it is sufficient for the analysis of the behavior and robustness of noisy decoders without requiring knowledge of a particular hardware implementation.
More accurate models will be considered in future works.


Note that some faulty functions cannot be decomposed as a deterministic mapping followed by the noise effect. 
For example, it can be verified that the faulty minimum function defined as
 \begin{equation}
  \tilde{\eta}_3 = \left\{
\begin{array}{rl}
  & \min(\mu_1,\mu_2) \text{ with probability } 1 - p  \\
  & \max(\mu_1,\mu_2) \text{ with probability } p
\end{array}
\right.
 \end{equation}
 does not satisfy~\eqref{eq:condind}.

\subsection{Particular Decoder Noise Models}
In the following, two particular noise models models that have been proposed in~\cite{ngassa14ITA} will be considered.
They are derived from the above decomposition model by specifying particular transition matrices $\Pi^{(c)},\Pi^{(v)},\Pi^{(a)}$ and will be considered for the noisy-DE analysis.

\subsubsection{Sign-Preserving error model}
The first model is called the Sign-Preserving (SP) model.
It has a SP property, meaning that noise is assumed to affect only the message amplitude, but not its sign.
Although this model is introduced for the purpose of asymptotic analysis, it is also a practical model, as protecting the sign can be realized at the hardware level by proper circuit design.
The probability transition matrices for the SP-Model can be constructed from a SP-transfer matrix defined as follows.
\begin{definition}\label{def:spmodel}
 The SP-transfer matrix $\Pi^{(\text{SP})}(p,s)$ is a matrix of size $(2s+1)\times(2s+1)$ such that
 \begin{align} \notag
  &  \Pi_{k,k}^{(\text{SP})}(p,s)  = 1 - p, ~~~ \Pi_{k,0}^{(\text{SP})}(p,s)  = \frac{p}{s},  ~~~ \Pi_{0,k}^{(\text{SP})}(p,s)  = \frac{p}{2s} \\ \notag
  & \Pi_{k,m}^{(\text{SP})}(p,s)  = \frac{p}{s}, \text{ for } m\neq k \neq 0, ~ \text{sign}(m) = \text{sign}(k) \\
  & \Pi_{k,k}^{(\text{SP})}(p,s) = 0 \text{, elsewhere} .
\end{align}
\end{definition}
According to this definition, a strictly positive message can be altered to only another positive message and the same holds for strictly negative messages.

The matrices $\Pi^{(c)}$, $\Pi^{(v)}$, and $\Pi^{(a)}$ can be now obtained from $\Pi^{(\text{SP})}$ as a template.
The noise level parameter at the output of $\Phic$ is given by the parameter $p_c$, and the corresponding probability transition matrix is given by $\Pi^{(c)} = \Pi^{(\text{SP})}(p_c,s)$.
In the same way, the noise level parameters at the output of $\Phi^{(v)}$ and $\Phi^{(a)}$ are denoted $p_v$ and $p_a$ respectively, and the corresponding probability transition matrices are given by $\Pi^{(v)} = \Pi^{(\text{SP})}(p_v,s)$ and $\Pi^{(a)} = \Pi^{(\text{SP})}(p_a,s')$.
In the following, the collection of hardware noise parameters will be denoted $\nu = (p_v,p_c,p_a)$.
The probability transition matrix $\Pi^{(a)}$ is of size $(2s'+1)\times(2s'+1)$ because the APP~\eqref{eq:Phi_app} is computed on the alphabet $\bar{\mathcal{M}} $ of size $(2s'+1)$.
It can be verified that if the deterministic mappings $\Phiv$, $\Phic$, $\Phia$, are symmetric in the sense of~\cite[Definition 1]{richardson01IT2}, then the SP-model gives symmetric faulty functions from conditions~\eqref{eq:randvnusym},~\eqref{eq:randcnusym},~\eqref{eq:randappsym}, in Definition~\ref{def:sym}.

\subsubsection{Full-Depth error model}
The second model is called the Full-Depth (FD) model.
This model is potentially more harmful than the SP-Model because the noise affects both the amplitude and the sign of the messages. 
However, it does not require hardware sign-protection any more.
The FD-transfer  matrix is defined as follows.
\begin{definition}\label{def:fdmodel}
 The FD-transfer matrix $\Pi^{(\text{FD})}(p,s)$ is a matrix of size $(2s+1)\times(2s+1)$ such that
 \begin{align} \notag
  & \Pi_{k,k}^{(\text{FD})}(p,s)  = 1 - p, \\
  &  \Pi_{k,m}^{(\text{FD})}(p,s)  = \frac{p}{s}, \text{ for } m\neq k .
\end{align}
\end{definition}
The FD-transfer Matrix defines a $(2s+1)$-ary symmetric model of parameter $p$.
The noise level parameters at the end of $\Phi_c$, $\Phi_v$, $\Phi_a$, are denoted as before $p_c$, $p_v$, $p_a$, respectively, and $\nu = (p_v,p_c,p_a)$.
The corresponding probability transition matrices are given by $\Pi^{(c)} = \Pi^{(\text{FD})}(p_c,s)$, $\Pi^{(v)} = \Pi^{(\text{FD})}(p_v,s)$, and $\Pi^{(a)} = \Pi^{(\text{FD})}(p_a,s')$.
It can be verified that if the deterministic mappings $\Phiv$, $\Phic$, $\Phia$, are symmetric in the sense of~\cite[Definition 1]{richardson01IT2}, then the FD-model gives symmetric faulty functions from the conditions~\eqref{eq:randvnusym},~\eqref{eq:randcnusym},~\eqref{eq:randappsym}, in Definition~\ref{def:sym}.

\section{Noisy Density Evolution}\label{sec:FaultyDE}
This section presents the noisy-DE recursion for asymptotic performance analysis of FAIDs on faulty hardware.
The DE~\cite{varshney2011performance} consists of expressing the Probability Mass Function (PMF) of the messages at successive iterations under the local independence assumption, that is the assumption that the messages coming to a node are independent.
As a result, the noisy-DE equations can be used to derive the error probability of the considered decoder as a function of the hardware noise parameters.
The noisy-DE analysis is valid on average over all possible LDPC code constructions, when infinite codeword length is considered.

In the following, we first discuss the all-zero codeword assumption which derives from the symmetry conditions of Definition~\ref{def:sym} and greatly simplifies the noisy-DE analysis.

\subsection{All-zero Codeword Assumption}
In~\cite{richardson01IT2}, it was shown that if the channel is output-symmetric, and the VNU and CNU functions are symmetric functions, the error probability of the decoder does not depend on the transmitted codeword.
From this codeword independence, one can compute the PMFs of the messages and the error probability of the decoder assuming that the all-zero codeword was transmitted.
The codeword independence was further extended in~\cite{ngassa14ITA,varshney2011performance} to the case of faulty decoders when the noise is introduced through symmetric error injection functions.
Unfortunately, the results of~\cite{varshney2011performance,ngassa14ITA} do not apply to our more general error models.
In particular, the proof technique of~\cite{varshney2011performance,ngassa14ITA} cannot be used when the noise is not introduced through deterministic error injection functions.
The following theorem thus restates the codeword independence for faulty functions described by the general error introduced in Section~\ref{sec:genmodel} and for the symmetry conditions of Definition~\ref{def:sym}.
\begin{te}\label{th:pe}
 Consider a linear code and a faulty decoder defined by a faulty VNU~\eqref{eq:sVNU}, a faulty CNU~\eqref{eq:sCNU}, and a faulty APP~\eqref{eq:sAPP}.
 Denote $P_e^{(\ell)}(\mathbf{x})$ the probability of error of the decoder at iteration $\ell$ conditioned on the fact that the codeword $\mathbf{x}$ was transmitted.
 If the transmission channel is symmetric in the sense of~\cite[Definition 1]{richardson01IT2} and if the faulty VNU, CNU, and APP are symmetric in the sense of Definition~\ref{def:sym}, then $P_e^{(\ell)}(\mathbf{x})$ does not depend on $\mathbf{x}$.
\end{te}

\begin{proof}
 See Appendix.
\end{proof}
Theorem~\ref{th:pe} states that for a symmetric transmission channel and symmetric faulty functions, the error probability of the decoder is independent of the transmitted codeword.
All the error models considered in the paper are symmetric and as a consequence, we will assume that the all-zero codeword was transmitted.
Note that when the decoder is not symmetric, DE can be performed from the results of~\cite{bennatan06IT,wang05IT}.
In this case, it is not possible anymore to assume that the all-zero codeword was transmitted, and the analysis becomes much more complex.

\subsection{Noisy-DE Equations}
In this section, we assume that the all-zero codeword was transmitted, and we express the PMFs of the messages at successive iterations.
The error probability of the decoder at a given iteration can then be computed from the PMFs of the messages at the considered iteration.
The analysis is presented for regular LDPC codes.
However, the generalization to irregular codes is straightforward. 

Let the $N_s$-tuple $\mathbf{q}^{(\ell)}$ denote the PMF of an outgoing message from a VN at $\ell$-th iteration.
In other words, the $\mvtoc$-th component $\pvtoc_{\mvtoc}^{(\ell)}$ of $\mathbf{\pvtoc}^{(\ell)}$ is the probability that the outgoing message takes the value $\mvtoc \in \mathcal{M}$.
Similarly, let $\mathbf{r}^{(\ell)}$ denote the PMF of an outgoing message from a CN.
The PMFs of noisy messages are represented by $\mathbf{\tilde{\pvtoc}}^{(\ell)}$ and $\mathbf{\tilde{\pctov}}^{(\ell)}$, respectively.
In the following, the noisy-DE recursion is expressed with respect to general probability transition matrices $\Pi^{(c)}$, $\Pi^{(v)}$, $\Pi^{(a)}$ .
To obtain the noisy-DE equations for a specific error model, it suffices to replace these general probability transition matrices with the ones corresponding to the considered model.

The density evolution is initialized with the PMF of the channel value
\[ \pvtoc_{-B}^{(0)}=1-\alpha  \hspace*{5mm} \pvtoc_{+B}^{(0)}=\alpha \hspace*{5mm} \pvtoc_{k}^{(0)}=0 \mbox{ elsewhere.} \]
Denote $\tilde{\mathbf{\pvtoc}}_{\boldsymbol{\mvtoc}}^{(\ell-1)}$ the $(d_c-1)$-tuple associated to $\boldsymbol{\mvtoc}$.
More precisely, if the $k$-th component of $\boldsymbol{\mvtoc}$ is given by $\mvtoc_k$, then the $k$-th component of $\tilde{\mathbf{\pvtoc}}_{\boldsymbol{\mvtoc}}^{(\ell-1)}$ is given by $\tilde{\pvtoc}_{\mvtoc_k}^{(\ell-1)}$.
The PMF $\mathbf{\pctov}^{(\ell)}$ of the output of the CNU is obtained from the expression of $\Phi_c$ as $\forall \mctov \in \mathcal{M}$,
\begin{equation}
\pctov_{\mctov}^{(\ell)} = \underset{\boldsymbol{\mvtoc} : \Phi_c(\boldsymbol{\mvtoc})=\mctov }{\sum} \prod \tilde{\mathbf{\pvtoc}}_{\boldsymbol{\mvtoc}}^{(\ell-1)}
\label{eq:DEchecknode}
\end{equation}
where the vector product operator is performed componentwise on vector elements.
The noisy PMF is then obtained directly in vector form as
\begin{equation}\label{eq:DEnoisychecknode}
 \mathbf{\tilde{\pctov}}^{(\ell)} = \Pi^{(c)} \mathbf{\pctov}^{(\ell)} .
\end{equation}
Denote $\tilde{\mathbf{\pctov}}_{\boldsymbol{\mctov}}^{(\ell)}$ the $(d_v-1)$-tuple associated to $\boldsymbol{\mctov}$.
The PMF $\mathbf{\pvtoc}^{(\ell)}$ of the output of the VNU is obtained from the expression of $\Phi_v$ as $\forall \mvtoc \in \mathcal{M}$,
\small
\begin{equation}
\pvtoc_{\mvtoc}^{(\ell)} = \underset{\boldsymbol{\mctov} : \Phi_v(\boldsymbol{\mctov},-B)=\mvtoc }{\sum} \pvtoc_{-B}^{(0)} \prod \tilde{\mathbf{\pctov}}_{\boldsymbol{\mctov}}^{(\ell)} ~~ + \underset{\boldsymbol{\mctov} : \Phi_v(\boldsymbol{\mctov},+B)=\mvtoc }{\sum} \pvtoc_{+B}^{(0)} \prod \tilde{\mathbf{\pctov}}_{\boldsymbol{\mctov}}^{(\ell)}
                    \label{eq:DEbitnode}
\end{equation}
\normalsize
and
\begin{equation}\label{eq:DEnoisybitnode}
 \mathbf{\tilde{\pvtoc}}^{(\ell)} =  \Pi^{(v)} \mathbf{\pvtoc}^{(\ell)}.
\end{equation}
Finally, applying the sequence of 4 equations \eqref{eq:DEchecknode}, \eqref{eq:DEnoisychecknode}, \eqref{eq:DEbitnode} and \eqref{eq:DEnoisybitnode} 
implements one recursion of the noisy-DE over the BSC channel.

The error probability of the decoder can be obtained from the above recursion and from the PMF of the messages at the end of the APP computation.
Denote $\tilde{\mathbf{\pctov}}_{\boldsymbol{\bar{\mctov}}}^{(\ell)}$ the $d_v$-tuple associated to $\boldsymbol{\bar{\mctov}}$, and denote $\mathbf{\pvtoc}_{\text{app}}^{(\ell)}$ and $\mathbf{\tilde{\pvtoc}}_{\text{app}}^{(\ell)}$ the respective noiseless and noisy PMFs of the messages at the output of the APP computation.
They can be expressed from~\eqref{eq:Phi_app} as $\forall \gamma \in \bar{\mathcal{M}}$,
\small
\begin{equation} \notag
\pvtoc_{\text{app},\gamma}^{(\ell)} =  \underset{\boldsymbol{\bar{\mctov}}: \Phi_a(\tilde{\boldsymbol{\mctov}}^{\star},-B)=\gamma }{\sum} \pvtoc_{-B}^{(0)} \prod \tilde{\mathbf{\pctov}}_{\boldsymbol{\bar{\mctov}}}^{(\ell)} ~~+ \underset{\boldsymbol{\bar{\mctov}} : \Phi_a(\tilde{\boldsymbol{\mctov}}^{\star},+B)=\gamma }{\sum} \pvtoc_{+B}^{(0)} \prod \tilde{\mathbf{\pctov}}_{\boldsymbol{\bar{\mctov}}}^{(\ell)}
\end{equation}\normalsize
and
\begin{equation} \label{eq:qappnoisy}
 \mathbf{\tilde{\pvtoc}}_{\text{app}}^{(\ell)} =  \Pi^{(a)} \mathbf{\pvtoc}_{\text{app}}^{(\ell)} .
\end{equation}
Finally, for a given $\alpha$ and hardware noise parameters $\nu = (p_v,p_c,p_a)$, the error probability at each iteration can be computed under the all-zero codeword assumption
as
\begin{equation}\label{eq:err_prob}
 P_{e,\nu}^{(\ell)}(\alpha) = \frac{1}{2}  \tilde{\pvtoc}_{\text{app},0}^{(\ell)} + \sum_{k < 0} \tilde{\pvtoc}_{\text{app},k}^{(\ell)} .
\end{equation}

Lower bounds on the error probability can be obtained as follows~\cite{Kameni14Coms}.
\begin{prop}\label{prop:lower-boundsFAID}
The following lower bounds hold at every iteration $\ell$ 
\begin{enumerate}
 \item For the SP model, $P_{e,\nu}^{(\ell)}(\alpha)  \geq \frac{1}{2s'} p_a$
 \item For the FD model, $P_{e,\nu}^{(\ell)}(\alpha)  \geq \frac{1}{2} p_a + \frac{p_a}{4s'}$
\end{enumerate}
\end{prop}
The term $s'$ appears in the two lower bounds because the APP~\eqref{eq:Phi_app} is computed on the alphabet $\bar{\mathcal{M}}$ of size $2s'+1$.	
%

The asymptotic error probability of an iterative decoder is the limit of $P_{e,\nu}^{(\ell)}(\alpha) $ when $\ell$ goes to infinity.
If the limit exists, let us denote $P_{e,\nu}^{(+\infty)}(\alpha) = \displaystyle \lim_{\ell \rightarrow +\infty} P_{e,\nu}^{(\ell)}(\alpha)$.
In the case of noiseless decoders ($p_v = p_c = p_a=0)$, the maximum channel parameter $\alpha$ such that $P_{e,\nu}^{(+\infty)}(\alpha) = 0$ is called the DE \emph{threshold} of the decoder~\cite{richardson01IT2}.
However, the condition $P_{e,\nu}^{(+\infty)}(\alpha) = 0$ cannot be reached in general for faulty decoders.
For instance, from Proposition~\ref{prop:lower-boundsFAID}, we see that the noise in the APP computation prevents the decoder from reaching a zero error probability.
Thus, the concept of iterative decoding threshold for faulty decoders has to be modified, and adapted to the fact that only very low asymptotic error probabilities, bounded away from zero, are achievable.
The following section recalls the definition of the functional threshold that was introduced in~\cite{Kameni14Coms,ngassa14ITA} to characterize the asymptotic behavior of faulty decoders.
We then analyze in details the properties of the functional threshold.

\section{Analysis of Convergence Behaviors of Faulty Decoders}\label{sec:th}

Varshney in~\cite{varshney2011performance} defines the {\it useful} region as the set of parameters $\alpha$ for which $P_{e,\nu}^{(+\infty)}(\alpha) < \alpha $.
The useful region indicates what are the faulty hardware and channel noise conditions that a decoder can tolerate to reduce the level of noise.
However, there are situations where the decoder can actually reduce the noise while still experiencing a high level of error probability.
As a consequence, the useful region does not predict which channel parameters lead to a low level of error probability.
Another threshold characterization has been proposed in~\cite{Balatsoukas14ComL,varshney2011performance}, where a constant value $\lambda$ is fixed and the target-BER threshold is defined as the maximum value of the channel parameter $\alpha$ such that $P_{e,\nu}^{(+\infty)}(\alpha) \leq \lambda $.
However, the target-BER definition has its limitations.
The choice of lambda is arbitrary, and the target-BER threshold does not capture an actual "threshold behavior", defined as a sharp transition between a low level and a high level of error probability.

Very recently, in~\cite{Kameni14Coms,ngassa14ITA}, another threshold definition referred to as the functional threshold has been proposed to detect the sharp transition between the two levels of error probability.
In this section, we first recall the functional threshold definition.
We then provide a new detailed analysis of the functional threshold behaviors and properties.
In particular, we point out the limitations of the functional threshold for the prediction of the asymptotic performance of faulty decoders.


\subsection{Functional Threshold Definition}
Here, we recall the  functional threshold definition introduced in~\cite{Kameni14Coms,ngassa14ITA}.
The functional threshold definition uses the Lipschitz constant of the function $\alpha \mapsto P_{e,\nu}^{(+\infty)}(\alpha)$ defined as
\begin{definition}
 Let $P_{e,\nu}^{(+\infty)}:I \rightarrow \mathbb{R}$ be a function defined on an interval $I\subseteq \mathbb{R}$.
The {\em Lipschitz constant} of $P_{e,\nu}^{(+\infty)}$ in $I$ is defined as
\begin{equation}
L\left(P_{e,\nu}^{(+\infty)}, I\right) = \sup_{\alpha\neq \beta \in I} \frac{\lvert P_{e,\nu}^{(+\infty)}(\alpha)-P_{e,\nu}^{(+\infty)}(\beta)\rvert}{\lvert \alpha-\beta \rvert} \in \mathbb{R}_{+} \cup \{+\infty\}
\end{equation} 
For $a\in I$ and $\delta > 0$, let $I_a(\delta) = I \cap (a-\delta, a+\delta)$.  
The {\em (local) Lipschitz constant} of $P_{e,\nu}^{(+\infty)}$ in $\alpha\in I$ is defined by:
\begin{equation}
L\left(P_{e,\nu}^{(+\infty)}, \alpha\right) = \inf_{\delta > 0} L\left(P_{e,\nu}^{(+\infty)}, I_\alpha(\delta)\right) \in \mathbb{R}_{+} \cup \{+\infty\}
\end{equation} 
\end{definition}

Note that if $\alpha$ is a discontinuity point of $P_{e,\nu}^{(+\infty)}$, then $L\left(P_{e,\nu}^{(+\infty)}, \alpha\right) = +\infty$.
On the opposite, if $P_{e,\nu}^{(+\infty)}$ is differentiable in $\alpha$, then the Lipschitz constant in $\alpha$ corresponds to the absolute value of the derivative.
Furthermore, if $L\left(P_{e,\nu}^{(+\infty)}, I\right) < +\infty$, then $P_{e,\nu}^{(+\infty)}$ is  uniformly continuous on $I$ and almost everywhere differentiable. 
In this case, $P_{e,\nu}^{(+\infty)}$ is said to be {\em Lipschitz continuous} on $I$.

The functional threshold is then defined as follows.
\begin{definition}\label{def:ft}
For given decoder noise parameters $\nu = (p_v,p_c,p_a)$ and a given channel parameter $\alpha$, the decoder is said to be {\em functional} if it satisfies the three conditions below
\begin{description}
\item[$(a)$] The function $x \mapsto P_{e,\nu}^{(+\infty)}(x)$ is defined on $[0,\alpha]$,
\item[$(b)$] $P_{e,\nu}^{(+\infty)}$ is Lipschitz continuous on $[0, \alpha]$, and
\item[$(c)$] $L\left(P_{e,\nu}^{(+\infty)}, x\right)$ is an increasing function of $x\in [0, \alpha]$.
\end{description}

Then the functional threshold $\bar{\alpha}$ is defined as
\begin{equation}
 \bar{\alpha} = \sup \{ \alpha \mid \mbox{conditions } (a), (b) \mbox{ and } (c) \mbox{ above are satisfied}\}
\end{equation}
\end{definition}

The function $P_{e,\nu}^{(+\infty)}(x)$ is defined provided that there exist a limit of $P_{e,\nu}^{(\ell)}(x) $ when $\ell$ goes to infinity.
Condition $(a)$ is required because $P_{e,\nu}^{(\ell)}(x)$ does not converge for some particular decoders and noise conditions, as shown in~\cite{ngassa14ITA}.

The functional threshold is defined as the transition between two parts of the curve representing $P_{e,\nu}^{(\ell)}(\alpha)$ with respect to $\alpha$.
The first part corresponds to the channel parameters leading to a low level of error probability, \emph{i.e.}, for which the decoder can correct most of the errors from the channel.
In the second part, the channel parameters lead to a high level of error probability, meaning that the decoder does not operate properly..
Note that there are two possibilities.
If $L\left(P_{e,\nu}^{(+\infty)}, \bar{\alpha}\right) = +\infty$, then $\bar{\alpha}$ is a discontinuity point of $P_{e,\nu}^{(+\infty)}$ and the transition between the two levels is sharp.
If $L\left(P_{e,\nu}^{(+\infty)}, \bar{\alpha}\right) < +\infty$, then $\bar{\alpha}$ is just an inflection point of $P_{e,\nu}^{(+\infty)}$ and the transition is smooth. 
Using the Lipschitz constant defined in this section, it is possible to characterize the type of transition for the error probability and discriminate between the two cases. We provide more details on our approach in the next section.

%

\subsection{Functional Threshold Interpretation}
\begin{figure}[t]
\begin{center}
  \subfigure[~]{ \includegraphics[width=.3\linewidth]{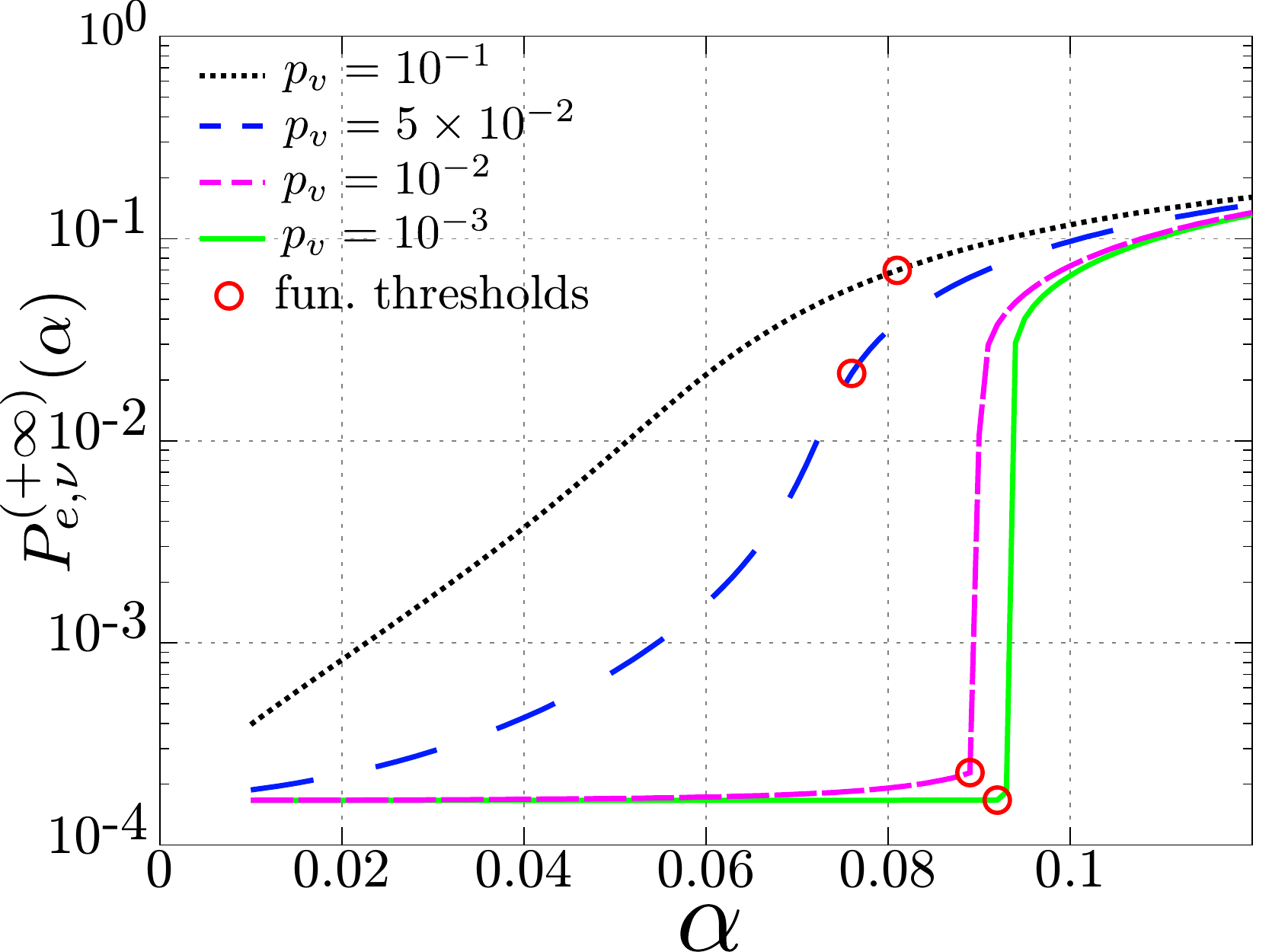}}
  \subfigure[~]{ \includegraphics[width=.3\linewidth]{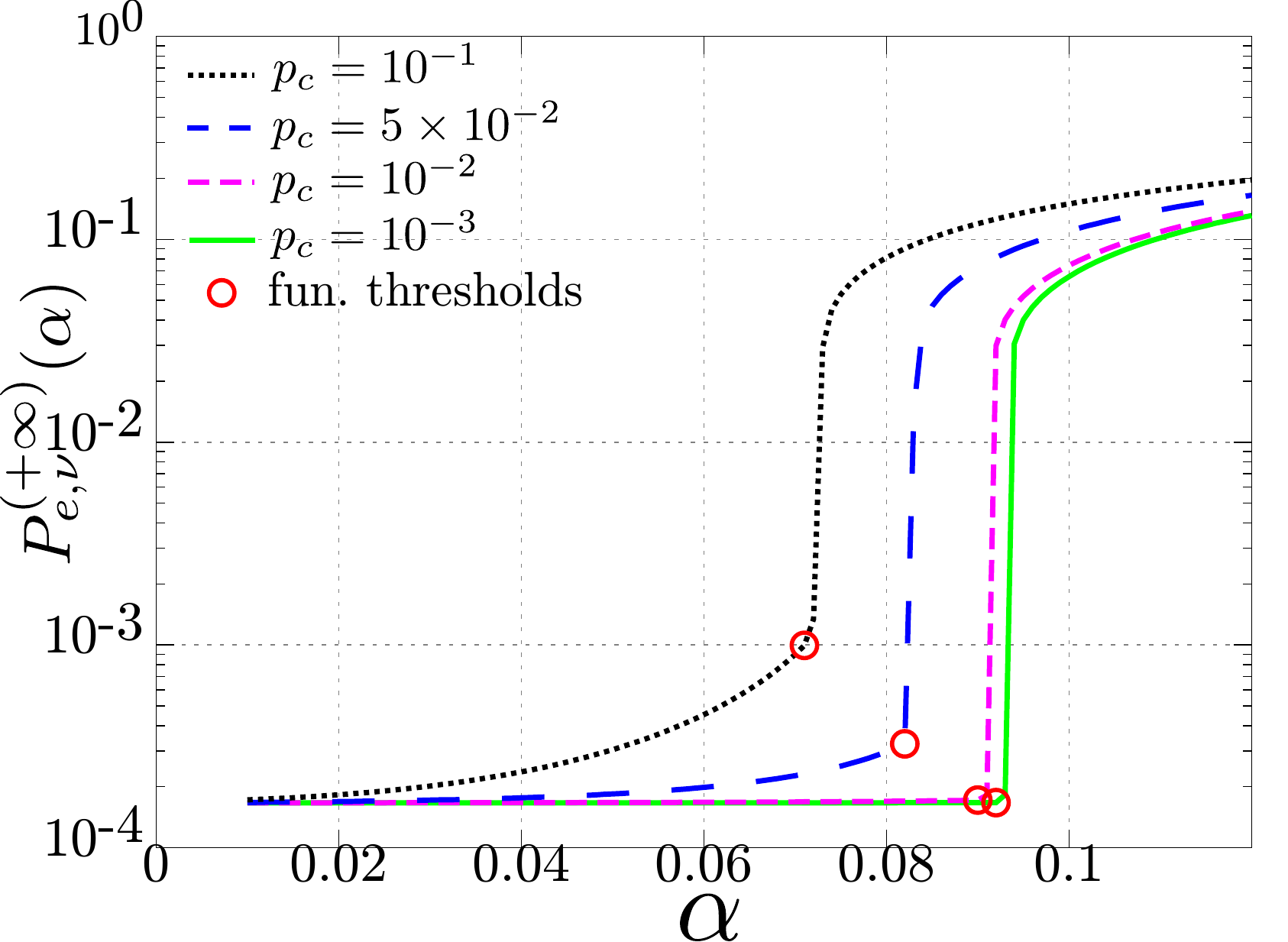}}
  \subfigure[~]{ \includegraphics[width=.3\linewidth]{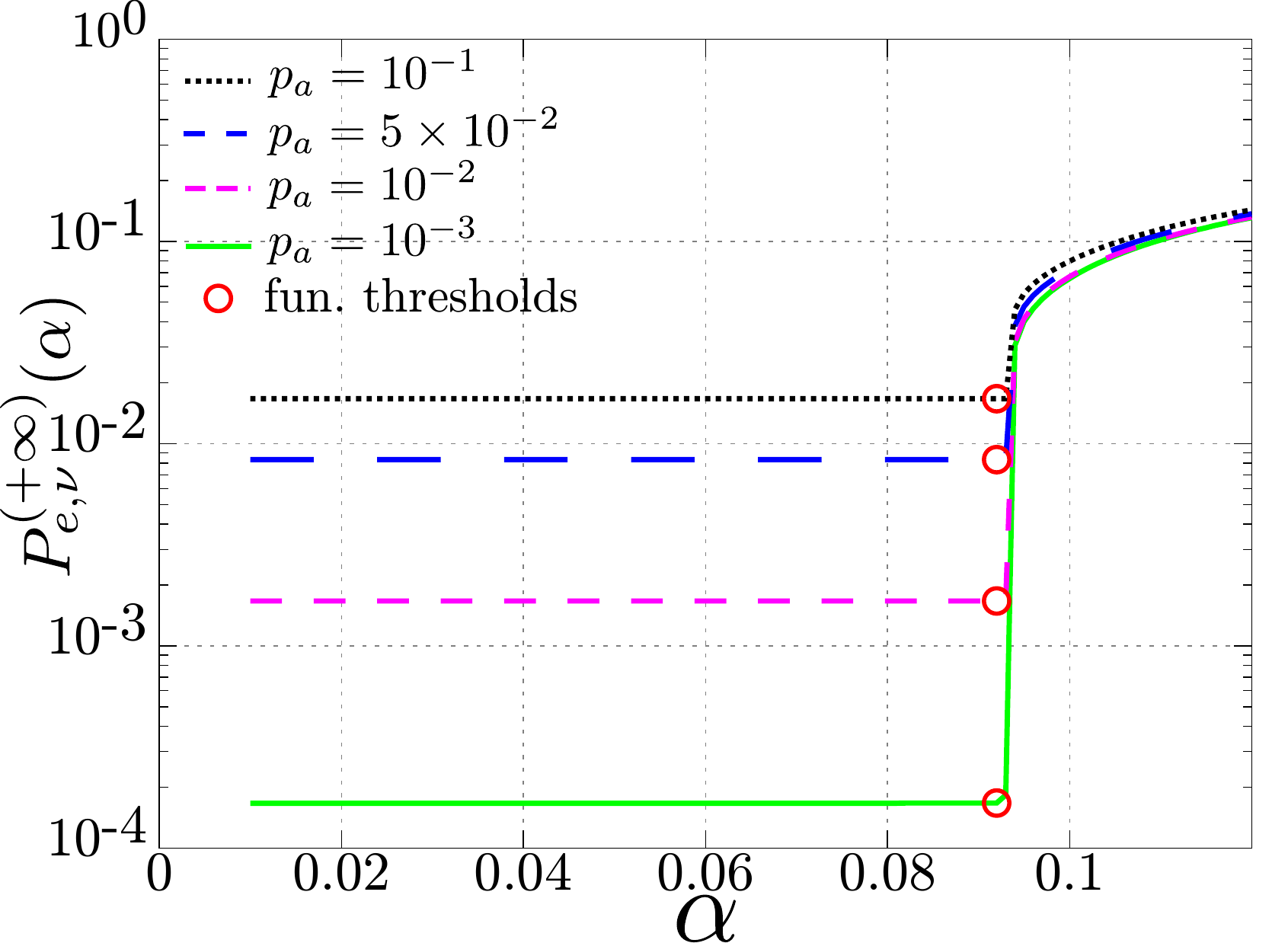}}
\end{center}
\caption{Asymptotic error probabilities for $(3,5)$ codes for the offset Min-Sum, for $B=1$, for the SP-Model, with (a) $p_c=10^{-3}$, $p_a=10^{-3}$, (b) $p_v=10^{-3}$, $p_a=10^{-3}$, (c)  $p_v=10^{-3}$, $p_c=10^{-3}$}
\label{fig:err_probs}
\end{figure}

\begin{figure}[t]
\begin{center}
  \subfigure[~]{ \includegraphics[width=.3\linewidth]{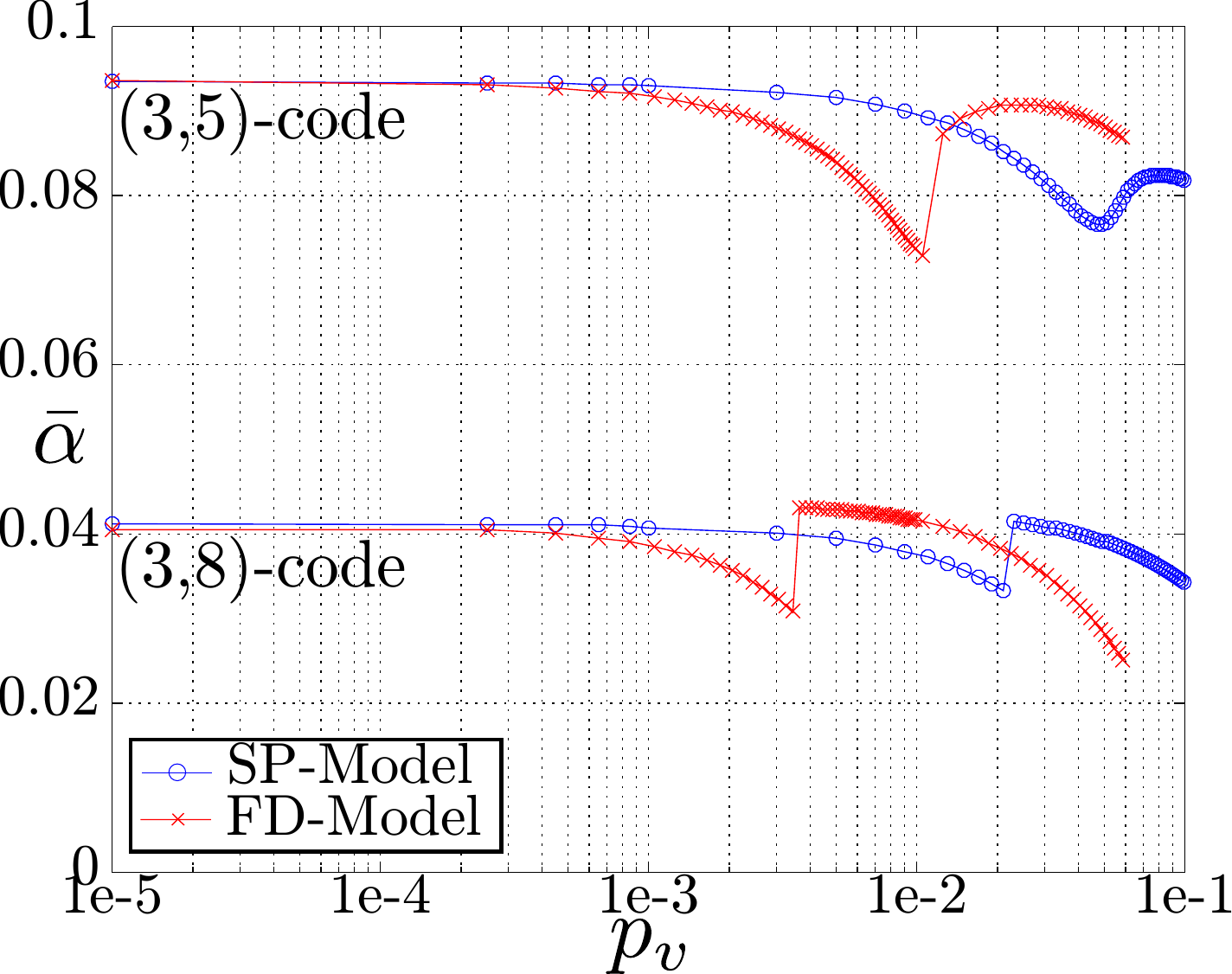}}
  \subfigure[~]{ \includegraphics[width=.3\linewidth]{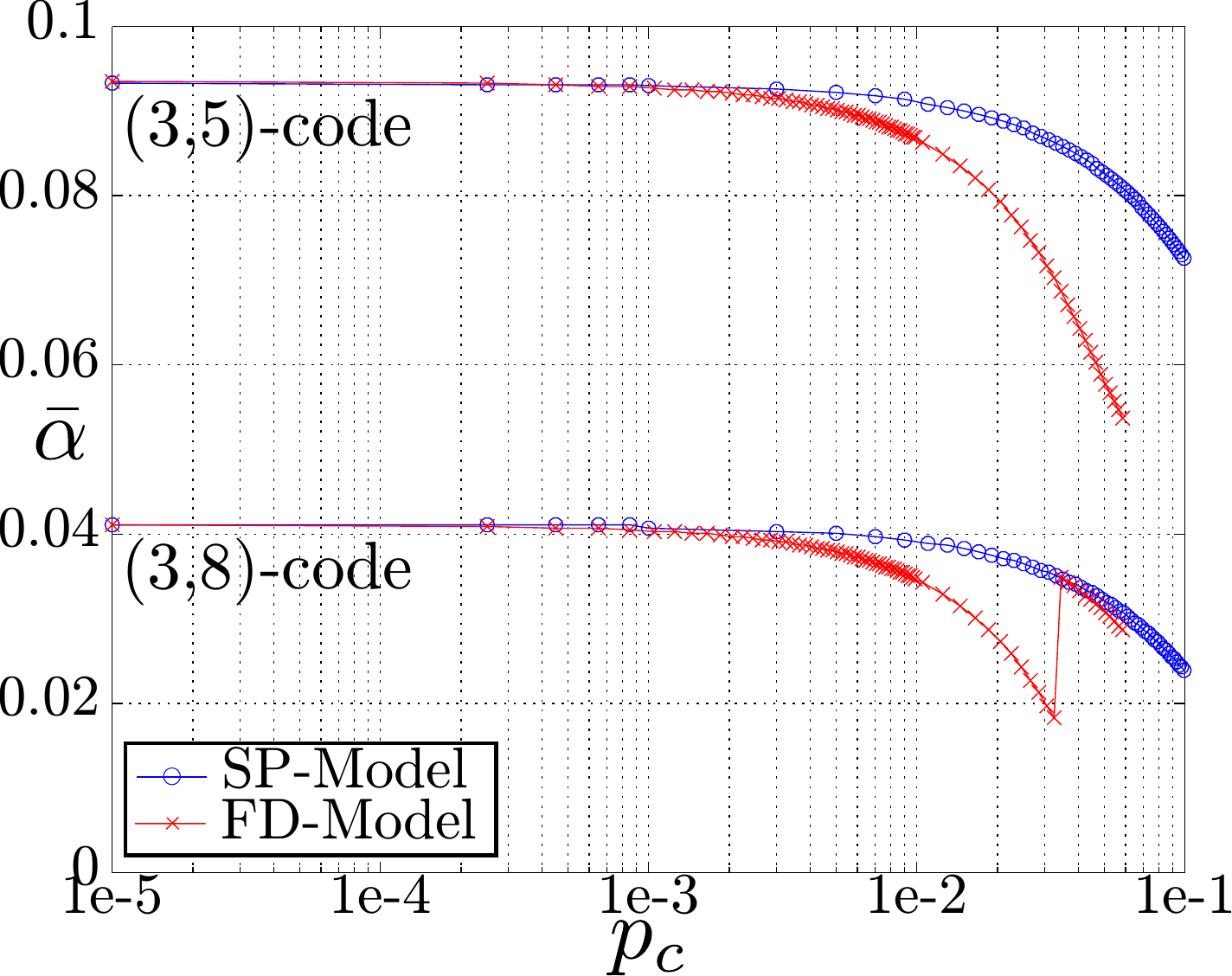}}
  \subfigure[~]{ \includegraphics[width=.33\linewidth]{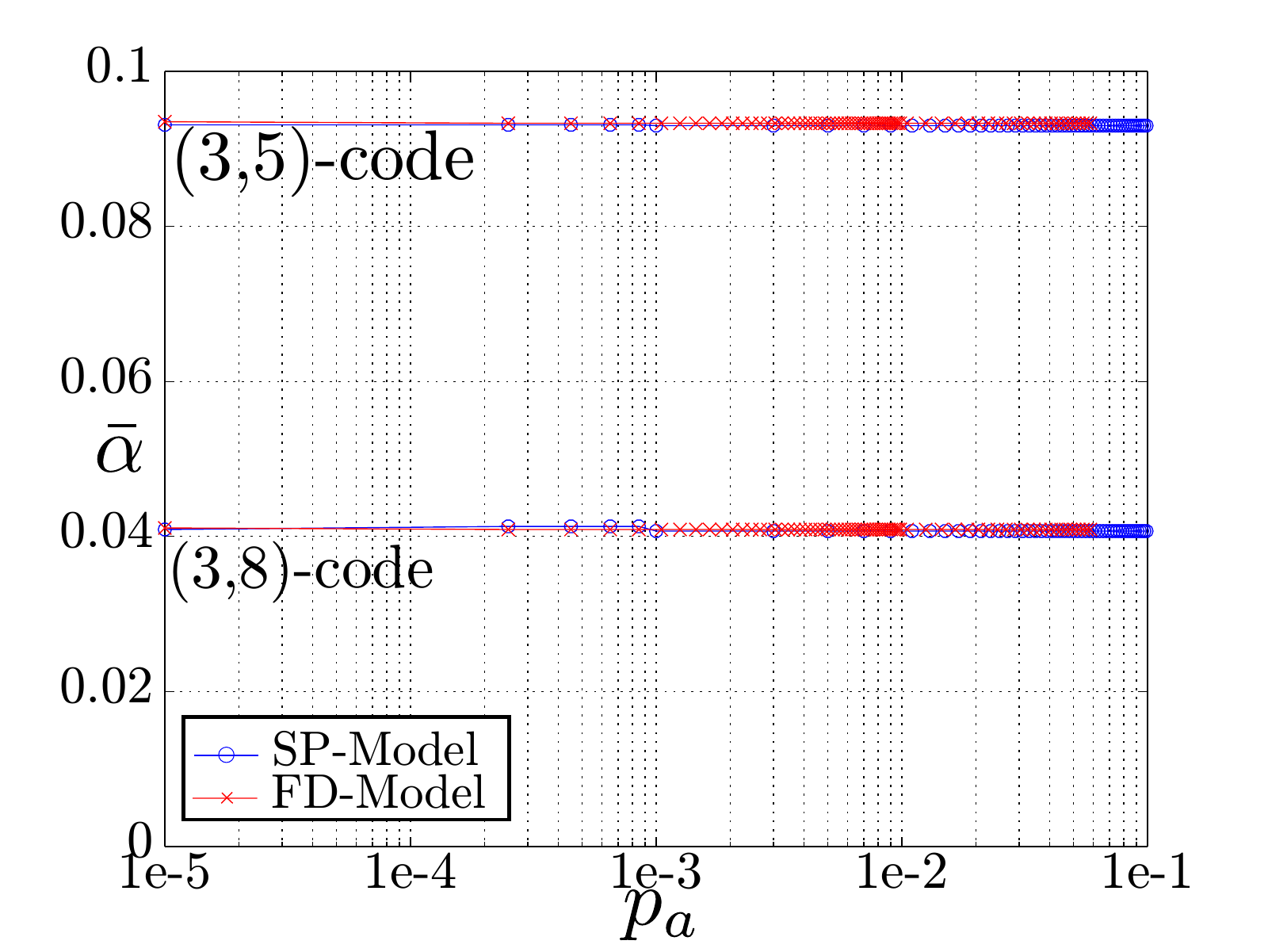}}
\end{center}
\caption{Functional regions for the offset min-sum, for $B=1$, (a) w.r.t. $p_v$, with $p_c= p_a=10^{-3}$ (SP-Model) and $p_c= p_a=10^{-4}$ (FD-Model), (b) w.r.t. $p_c$, with $p_v = p_a=10^{-3}$ (SP-Model) and $p_v= p_a=10^{-4}$ (FD-Model), (c) w.r.t. $p_a$, with $p_v = p_c=10^{-3}$ (SP-Model) and $p_v= p_c=10^{-4}$ (FD-Model)}
\label{fig:ftregions}
\end{figure}

As opposed to the work presented in~\cite{Kameni14Coms,ngassa14ITA}, where the functional threshold was introduced only to predict the asymptotic performance of the faulty Min-Sum decoder, our goal is to use the functional threshold as a tool to discriminate between different FAIDs and design faulty decoders which are robust to faulty hardware. 
In order to do so, we need a precise understanding of the behaviors and the limits of the functional threshold.
We present the analysis for regular $d_v=3$ LDPC codes, and for the offset Min-Sum decoder~\cite{Chen05Com} interpreted as a FAID.
Table~\ref{tab:3biMinSum} gives the LUT of the VNU of the $7$-level offset Min-Sum decoder considered for the analysis.

Fig.~\ref{fig:err_probs} (a) represents the asymptotic error probability $P_{e,\nu}^{(+\infty)}(\alpha)$ with respect to $\alpha$ for several values of $p_v$ for the SP-Model with $p_c = p_a = 10^{-3}$.
The circled points represent the positions of the functional thresholds obtained from Definition~\ref{def:ft}.
When $p_v$ is low, the threshold is given by the discontinuity point of the error probability curve.
But when $p_v$ becomes too high, there is no discontinuity point anymore, and the functional threshold is given by the inflection point of the curve.
However, the inflection point does not predict accurately which channel parameters lead to a low level of error probability.
Fig.~\ref{fig:err_probs} (b) represents $P_{e,\nu}^{(+\infty)}(\alpha)$ for several values of $p_c$ with $p_v = p_a = 10^{-3}$.
In all the considered cases, the functional threshold is given by the discontinuity point of the error probability curve.
Fig.~\ref{fig:err_probs} (c) represents $P_{e,\nu}^{(+\infty)}(\alpha)$ for several values of $p_a$ with $p_v = p_c = 10^{-3}$.
In this case, not only the functional threshold is always given by the discontinuity point of the error probability curve, but the position of the functional threshold position does not seem to depend on the value of $p_a$.

Fig.~\ref{fig:ftregions} (a) shows the functional thresholds $\bar{\alpha}$ as a function of the hardware noise parameter at the VNU, $p_v$.
For the SP-Model, we consider $p_c = p_a = 10^{-3}$, and for the FD-Model, $p_c = p_a = 10^{-4}$.
When $p_v$ is small, the value of $\bar{\alpha}$ decreases with increasing $p_v$. 
But when $p_v$ becomes too large, we observe an unexpected jump in the $\bar{\alpha}$ values.
The curve part at the right of the jump corresponds to the values $p_v$ for which the functional threshold is given by the inflection point of the error probability curve.
This confirms that when $p_v$ is too large, the functional threshold does not predict accurately which channel parameters lead to a low level of error probability.
Fig.~\ref{fig:ftregions} (b) shows the $\bar{\alpha}$ values as a function of $p_c$.
For the $(3,8)$-code and the FD-Model, we observe that when $p_c$ becomes too large, the functional threshold also fails at predicting the convergence behavior of the faulty decoder.
Finally, Fig.~\ref{fig:ftregions} (c) shows the $\bar{\alpha}$ values as a function of $p_a$.
It confirms that the functional threshold value does not depend on $p_a$.u
This is expected, because the APP computation does not affect the iterative decoding process.
As a consequence, the faulty APP computation only adds noise in the final codeword estimate, but does not make the decoding process fail.

We have seen that when the hardware noise is too high, it leads to a non-standard asymptotic behavior of the decoder in which the functional threshold does not predict accurately the convergence behavior of the faulty decoder.
That is why we modify the functional threshold definition as follows.
\begin{definition}\label{def:admit2}
Denote $\alpha^{\star}$ the functional threshold value obtained from Definition~\ref{def:ft}.
The functional threshold value is restated by setting its value to $\bar{\alpha}$ defined as

\begin{equation}
 \bar{\alpha} = \left\{ \begin{array}{ll}
                        \alpha^{\star} & \text{ if } L\left(P_{e,\nu}^{(+\infty)}, \alpha^{\star}\right)  = +\infty, \\
                        0 & \text{ if } L\left(P_{e,\nu}^{(+\infty)}, \alpha^{\star}\right)  < +\infty .
                       \end{array}
 \right.
\end{equation}

\end{definition}
Definition~\ref{def:admit2} eliminates the decoder noise values which lead to non-desirable behavior of the decoder.
The functional threshold  of Definition~\ref{def:admit2} identifies the channel parameters $\alpha$ which lead to a low level of asymptotic error probability and predicts accurately the convergence behavior of the faulty decoders. 
In this case, the functional threshold can be used as a criterion for the performance comparison of noisy FAIDs.
This criterion will be used in the following for the comparison of FAIDs performance and for the design of robust decoders.

\section{Design of FAIDs Robust to Faulty Hardware}\label{sec:selection}
Based on noisy-DE recursion and on the functional threshold definition, we now propose a method for the design of decoders robust to transient noise introduced by the faulty hardware.
In Section~\ref{sec:FaultyFAID}, we have seen that the FAID framework enables to define a large collection of VNU mappings $\Phi_v$ and thus a large collection of decoders.
The choice of the VNU mapping gives a degree of freedom for optimizing the decoder for a specific constraint.
In~\cite{Planjery_IEEETransCommun_2013}, FAIDs were optimized for low error flor.
Here, we want to optimize FAIDs for robustness to noise introduced by the faulty hardware.

\begin{figure}[t]
\begin{center}
  \subfigure[~]{ \includegraphics[width=.31\linewidth]{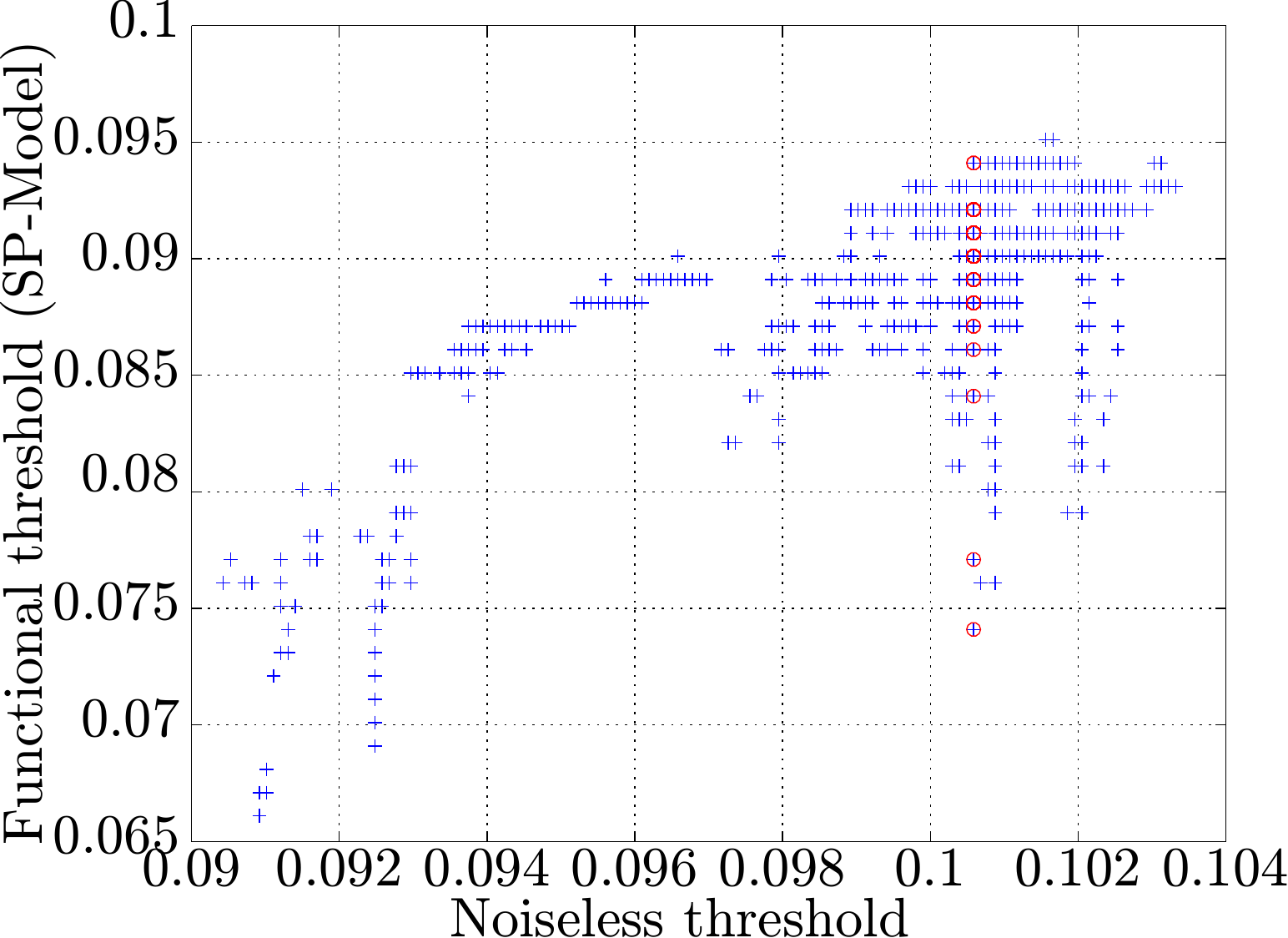}}
  \subfigure[~]{ \includegraphics[width=.31\linewidth]{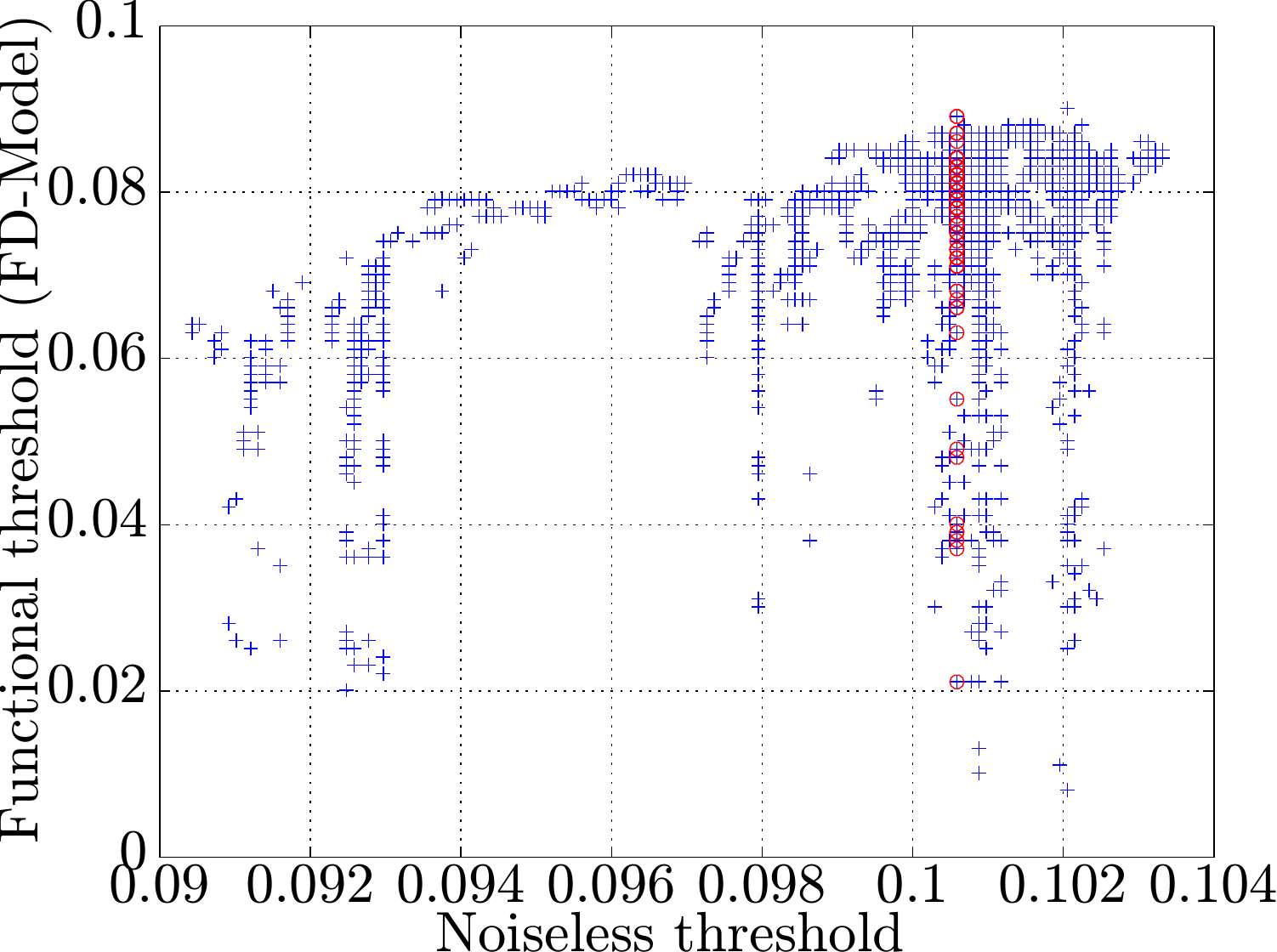}}
  \subfigure[~]{ \includegraphics[width=.31\linewidth]{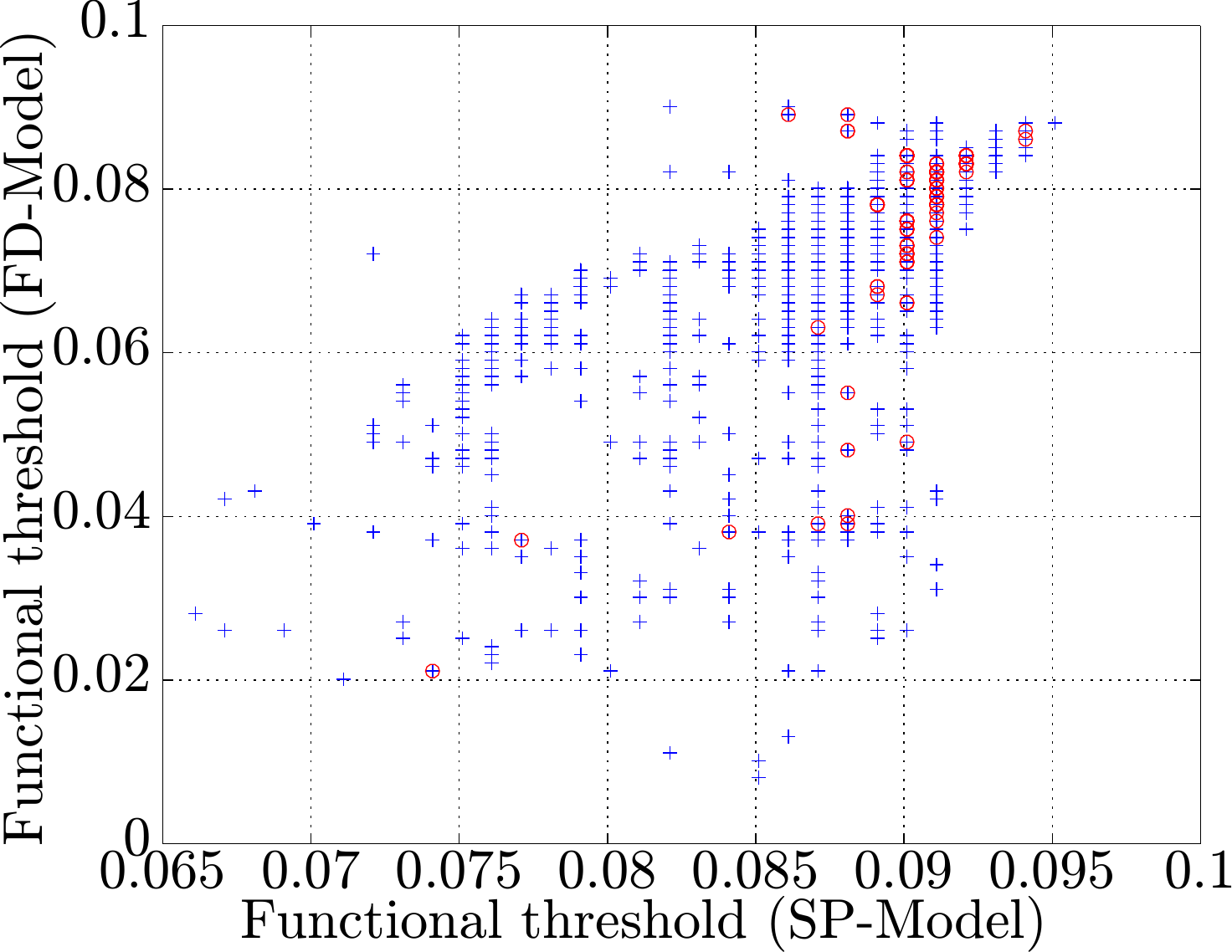}}
\end{center}
\caption{(a) Noiseless thresholds vs functional thresholds for the SP-Model ($p_v = p_c= p_a = 10^{-2}$) , (b) Noiseless thresholds vs functional thresholds for the FD-Model ($p_v=p_c=p_a=5 \times 10^{-3}$)   (c) Functional thresholds for the SP-Model ($p_v = p_c= p_a = 10^{-2}$) vs functional thresholds for the FD-Model ($p_v=p_c=p_a=5 \times 10^{-3}$) }

\label{fig:thres}
\end{figure}

For message alphabet size $N_s=7$, the number of possible FAIDs is equal to $530\,803\,988$, which is too large for a systematic analysis.
Instead, we rely on previous work on FAIDs, and start with a collection of $N_D=5291$ FAIDs which correspond to column-weight tree codes selected from the trapping sets analysis presented in~\cite{Planjery_IEEETransCommun_2013}. 
As a result of this selection process, each of the $N_D$ FAIDs have both good noiseless threshold, and good performance in the error floor.
We now perform a noisy-DE analysis on this set by computing, for each of the $N_D$ FAIDs, the value of their functional threshold.

As an illustration, Fig.~\ref{fig:thres} (a) and (b) represent the functional thresholds with respect to the noiseless thresholds.
For the SP-Model, the functional thresholds are computed for $p_v = p_c= p_a = 10^{-2}$, and for the FD-Model, $p_v = p_c= p_a = 5\times 10^{-3}$.
Although all the considered decoders have good noiseless threshold (between $0.09$ and $0.104$), a wide range of behaviors can be observed when the decoder is faulty.
Indeed, for the SP-Model, the functional threshold values are between $0.065 $ and $0.095$, thus illustrating the existence of both robust and non-robust decoders.
In particular, even decoders with approximately the same noiseless threshold value (e.g. around $0.101$) can exhibit different robustness.
This is even more pronounced for the FD-Model, for which the functional threshold values are between $0.01$ and $0.085$.
These observations illustrate the importance of selecting robust decoders to operate on faulty hardware and that a noiseless analysis is not sufficient to reach any useful conclusion.

We did also a performance comparison with noisy-DE and different error models, and Fig.~\ref{fig:thres} (c) represents the functional thresholds obtained for the FD-Model (for $p_v = p_c= p_a = 5\times 10^{-3}$) with respect to the functional thresholds obtained for the SP-Model (for $p_v = p_c= p_a = 10^{-2}$).
In this case also a large variety of behaviors can be observed.
Indeed, only a small number of decoders are robust to both error models,  while some of them are robust only to the SP-Model, and some others only to the FD-Model.
This suggests that robustness to different error models may require different decoders.

Following these observations, we have selected four decoders from the set of $N_D$ FAIDs.
The first two ones denoted $\Phi^{(v,\text{SP})}_{\mbox{\tiny robust}}$ and $\Phi^{(v,\text{FD})}_{\mbox{\tiny robust}}$ are the decoderd have been selected such as to minimize discrepancy between noiseless and functional thresholds, for the SP-Model and the FD-Model respectively.
Two other FAIDs $\Phi^{(v,\text{SP})}_{\mbox{\tiny non-robust}}$ and $\Phi^{(v,\text{FD})}_{\mbox{\tiny non-robust}}$ are selected to maximize the difference between noiseless and functional thresholds respectively for the SP-Model and for the FD-Model.
The LUTs of $\Phi^{(v,\text{SP})}_{\mbox{\tiny robust}}$ and $\Phi^{(v,\text{SP})}_{\mbox{\tiny non-robust}}$ are given in Table~\ref{tab:FAID_2075842} and Table~\ref{tab:FAID_941001}, and the LUTs of  $\Phi^{(v,\text{FD})}_{\mbox{\tiny robust}}$ and $\Phi_{v,\text{FD}}^{\mbox{\tiny (non-robust)}}$  are given in Table~\ref{tab:FAID_3161} and Table~\ref{tab:FAID_3888}.
The four decoders will be considered in the following section to validate the asymptotic noisy-DE results with finite-length simulations.

\begin{table*}[t]
\parbox{.48\linewidth}{%
\caption{FAID rule $\Phi^{(v,\text{SP})}_{\mbox{\tiny robust}}$ robust to the faulty Hardware (SP-Model)}
\centering 
\resizebox{8cm}{!}{
	\begin{tabular}{|c||c|c|c|c|c|c|c|}
	\hline
		\boldmath $m_{1}/m_{2}$         & \boldmath$-L_3$       &\boldmath$-L_2$        & \boldmath$-L_1$        & \boldmath $0$        & \boldmath$+L_1$    & \boldmath $+L_2$ & \boldmath $+L_3$\\ \hline\hline
		\boldmath$-L_3$                 &       $-L_3$    &       $-L_3$    &       $-L_3$    &       $-L_3$     &      $-L_3$     &      $-L_2$     &      $ 0$  \\ \hline
		\boldmath$-L_2$                 &       $-L_3$    &       $-L_3$    &       $-L_3$    &       $-L_3$     &      $-L_2$     &      $-L_2$     &      $L_1$  \\ \hline
		\boldmath$-L_1$                 &       $-L_3$    &       $-L_3$    &       $-L_3$    &       $-L_2$     &      $-L_1$     &      $-L_1$     &      $L_1$  \\ \hline
		\boldmath $0$                   &       $-L_3$    &       $-L_3$    &       $-L_2$    &       $-L_1$     &      $-L_1$     &      $ 0$     &      $L_1$  \\ \hline
		\boldmath $+L_1$                &       $-L_3$    &       $-L_2$    &       $-L_1$    &       $-L_1$     &      $ 0$     &      $L_1$     &      $L_2$  \\ \hline
		\boldmath $+L_2$                &       $-L_2$    &       $-L_2$    &       $-L_1$    &       $0$     &      $L_1$     &      $L_2$     &      $L_2$  \\ \hline
		\boldmath $+L_3$                &       $0$     &       $L_1$    &       $L_1$    &       $L_1$     &      $L_2$     &      $L_2$     &      $L_3$  \\ \hline
	\end{tabular}}
\label{tab:FAID_2075842}}%
\hfill
\parbox{.48\linewidth}{%
\caption{FAID rule $\Phi^{(v,\text{SP})}_{\mbox{\tiny non-robust}}$ not robust to faulty Hardware (SP-Model)}
\resizebox{8cm}{!}{
	\begin{tabular}{|c||c|c|c|c|c|c|c|}
	\hline
		\boldmath $m_{1}/m_{2}$         & \boldmath$-L_3$       &\boldmath$-L_2$        & \boldmath$-L_1$        & \boldmath $ 0 $        & \boldmath$+L_1$    & \boldmath $+L_2$ & \boldmath $+L_3$\\ \hline\hline
		\boldmath$-L_3$                 &       $-L_3$    &       $-L_3$    &       $-L_3$    &       $-L_3$     &      $-L_3$     &      $-L_3$     &      $ 0$  \\ \hline
		\boldmath$-L_2$                 &       $-L_3$    &       $-L_3$    &       $-L_3$    &       $-L_3$     &      $-L_2$     &      $ 0$     &      $L_2$  \\ \hline
		\boldmath$-L_1$                 &       $-L_3$    &       $-L_3$    &       $-L_2$    &       $-L_2$     &      $-L_1$     &      $ 0$     &      $L_2$  \\ \hline
		\boldmath $0$                   &       $-L_3$    &       $-L_3$    &       $-L_2$    &       $-L_1$     &      $ 0$     &      $L_1$     &      $L_3$  \\ \hline
		\boldmath $+L_1$                &       $-L_3$    &       $-L_2$    &       $-L_1$    &       $ 0$     &      $ 0$     &      $L_1$     &      $L_3$  \\ \hline
		\boldmath $+L_2$                &       $-L_3$    &       $0$     &       $0$     &       $L_1$     &      $L_1$     &      $L_1$     &      $L_3$  \\ \hline
		\boldmath $+L_3$                &       $0$    &       $L_2$    &       $L_2$    &       $L_3$     &      $L_3$     &      $L_3$     &      $L_3$  \\ \hline
	\end{tabular}}
\label{tab:FAID_941001}}%
\end{table*}

\begin{table*}[t]
\parbox{.48\linewidth}{%
\caption{FAID rule $\Phi^{(v,\text{FD})}_{\mbox{\tiny (robust)}}$ robust to the faulty Hardware (FD-Model)}
\centering 
\resizebox{8cm}{!}{
	\begin{tabular}{|c||c|c|c|c|c|c|c|}
	\hline
		\boldmath $m_{1}/m_{2}$         & \boldmath$-L_3$       &\boldmath$-L_2$        & \boldmath$-L_1$        & \boldmath $0$        & \boldmath$+L_1$    & \boldmath $+L_2$ & \boldmath $+L_3$\\ \hline\hline
		\boldmath$-L_3$                 &       $-L_3$    &       $-L_3$    &       $-L_3$    &       $-L_3$     &      $-L_3$     &      $-L_1$     &      $ 0$  \\ \hline
		\boldmath$-L_2$                 &       $-L_3$    &       $-L_3$    &       $-L_3$    &       $-L_3$     &      $-L_1$     &      $-L_1$     &      $L_2$  \\ \hline
		\boldmath$-L_1$                 &       $-L_3$    &       $-L_3$    &       $-L_2$    &       $-L_2$     &      $-L_1$     &      $0$     &      $L_2$  \\ \hline
		\boldmath $0$                   &       $-L_3$    &       $-L_3$    &       $-L_2$    &       $-L_1$     &      $0$     &      $ 0$     &      $L_3$  \\ \hline
		\boldmath $+L_1$                &       $-L_3$    &       $-L_1$    &       $-L_1$    &       $0$     &      $ 0$     &      $L_1$     &      $L_3$  \\ \hline
		\boldmath $+L_2$                &       $-L_1$    &       $-L_1$    &       $0$    &       $0$     &      $L_1$     &      $L_1$     &      $L_3$  \\ \hline
		\boldmath $+L_3$                &       $0$     &       $L_2$    &       $L_2$    &       $L_3$     &      $L_3$     &      $L_3$     &      $L_3$  \\ \hline
	\end{tabular}}
\label{tab:FAID_3161}}%
\hfill
\parbox{.48\linewidth}{%
\caption{FAID rule $\Phi^{(v,\text{FD})}_{\mbox{\tiny non-robust}}$ not robust to faulty Hardware (FD-Model)}
\resizebox{8cm}{!}{
	\begin{tabular}{|c||c|c|c|c|c|c|c|}
	\hline
		\boldmath $m_{1}/m_{2}$         & \boldmath$-L_3$       &\boldmath$-L_2$        & \boldmath$-L_1$        & \boldmath $ 0 $        & \boldmath$+L_1$    & \boldmath $+L_2$ & \boldmath $+L_3$\\ \hline\hline
		\boldmath$-L_3$                 &       $-L_3$    &       $-L_3$    &       $-L_3$    &       $-L_3$     &      $-L_2$     &      $-L_2$     &      $ 0$  \\ \hline
		\boldmath$-L_2$                 &       $-L_3$    &       $-L_3$    &       $-L_3$    &       $-L_3$     &      $-L_2$     &      $ -L_1$     &      $L_2$  \\ \hline
		\boldmath$-L_1$                 &       $-L_3$    &       $-L_3$    &       $-L_2$    &       $-L_2$     &      $-L_1$     &      $ 0$     &      $L_2$  \\ \hline
		\boldmath $0$                   &       $-L_3$    &       $-L_3$    &       $-L_2$    &       $-L_1$     &      $ 0$     &      $0$     &      $L_3$  \\ \hline
		\boldmath $+L_1$                &       $-L_2$    &       $-L_2$    &       $-L_1$    &       $ 0$     &      $ 0$     &      $L_1$     &      $L_3$  \\ \hline
		\boldmath $+L_2$                &       $-L_2$    &       $-L_1$     &       $0$     &       $0$     &      $L_1$     &      $L_1$     &      $L_3$  \\ \hline
		\boldmath $+L_3$                &       $0$    &       $L_2$    &       $L_2$    &       $L_3$     &      $L_3$     &      $L_3$     &      $L_3$  \\ \hline
	\end{tabular}}
\label{tab:FAID_3888}}%
\end{table*}

\section{Finite Length Simulations Results}\label{sec:results}
This section gives finite-length simulation results with the FAIDs $\Phi^{(v,\text{SP})}_{\mbox{\tiny robust}}$, $\Phi^{(v,\text{SP})}_{\mbox{\tiny non-robust}}$, $\Phi^{(v,\text{FD})}_{\mbox{\tiny robust}}$, and $\Phi^{(v,\text{FD})}_{\mbox{\tiny non-robust}}$ that have been identified by the noisy-DE analysis.
%
For the sake of comparison, a fifth decoder denoted $\Phi^{(v)}_{\mbox{\tiny (opt)}}$ (Table~\ref{tab:mapping}) will also be considered.
$\Phi^{(v)}_{\mbox{\tiny opt}}$ has been optimized in~\cite{Planjery_IEEETransCommun_2013} for noiseless decoding with low error floor.
In our simulations, the number of iterations is set to $100$ and we consider the $(155,93)$ Tanner code with degrees $(d_v=3,d_c=5)$ given in~\cite{tanner2004ldpc}.

Fig.~\ref{fig:flnoisy} (a) represents the Bit Error Rates (BER) with respect to channel parameter $\alpha$ and for the SP-Model.
In the case of noiseless decoding, as $\Phi^{(v)}_{\mbox{\tiny opt}}$ has been optimized for low error floor, it performs better, as expected, than $\Phi^{(v,\text{SP})}_{\mbox{\tiny robust}}$ and $\Phi^{(v,\text{SP})}_{\mbox{\tiny non-robust}}$. 
But as $\Phi^{(v,\text{SP})}_{\mbox{\tiny robust}}$ and $\Phi^{(v,\text{SP})}_{\mbox{\tiny non-robust}}$ belong to a predetermined set of good FAID decoders, they also have good performance in the noiseless case.

We now discuss the faulty decoding case.
For the SP-Model, we fix $p_v = p_c = p_a = 0.05$, and for the FD-Model,  $p_v = p_c = p_a = 0.02$.
We first see that the lower bound conditions of Proposition~\ref{prop:lower-boundsFAID} are not satisfied here.
Indeed, in our simulations, we considered an early stopping criterion, which halts the decoding process when the sequence estimated by the APP block is a codeword, while the  results of Proposition~\ref{prop:lower-boundsFAID} consider the averaged error probabilities at a fixed iteration number, and thus do not take into account the stopping criterion. 
We then see that the results are in compliance with the conclusions of the functional thresholds analysis.
Indeed, when the decoder is faulty, $\Phi^{(v,\text{SP})}_{\mbox{\tiny robust}}$ performs better than $\Phi^{(v)}_{\mbox{\tiny opt}}$ while $\Phi^{(v,\text{SP})}_{\mbox{\tiny non-robust}}$ has a significant performance loss compared to the two other decoders.
From Fig.~\ref{fig:flnoisy} (b) we see that the same holds for the FD-Model in which case the error correction performance of the faulty decoders are much worse than for the SP-Model.
The FD-Model makes decoders less robust to noise than the SP-Model, because with the FD-Model, not only the amplitudes, but also the signs of the messages can be corrupted by the noise.
In particular, the non-robust decoder $\Phi^{(v,\text{FD})}_{\mbox{\tiny non-robust}}$ performs extremely poorly.

We now comment the results of Fig.~\ref{fig:compspfd}.
The code and decoder noise parameters are the same as before.
In Fig.~\ref{fig:compspfd}, the FD-Model with $p_v = p_c = p_a = 0.02$ is applied to $\Phi^{(v,\text{SP})}_{\mbox{\tiny robust}}$ and $\Phi^{(v,\text{FD})}_{\mbox{\tiny robust}}$, and the SP-Model with $p_v = p_c = p_a = 0.05$ is also applied to $\Phi^{(v,\text{SP})}_{\mbox{\tiny robust}}$ and $\Phi^{(v,\text{FD})}_{\mbox{\tiny robust}}$.
We see that $\Phi^{(v,\text{SP})}_{\mbox{\tiny robust}}$ is robust for the SP-Model but not-robust for the FD-Model and that $\Phi^{(v,\text{FD})}_{\mbox{\tiny robust}}$ is robust for the FD-Model but not-robust for the SP-Model.
These results are in compliance with the asymptotic analysis of Section~\ref{sec:selection} which shows that some decoders that are robust for one model are not necessarily robust for the other one.

To conclude, the finite-length simulations confirm that the functional threshold can be used to predict the performance of faulty decoders.
Both the asymptotic analysis and the finite-length results demonstrate the existence of robust and non-robust decoders.
They both illustrate the importance of designing robust decoders for faulty hardware and show that the design of robust decoders is dependant on the hardware error model.

\begin{figure}[t]
\begin{center}
  \subfigure[~]{ \includegraphics[width=.48\linewidth]{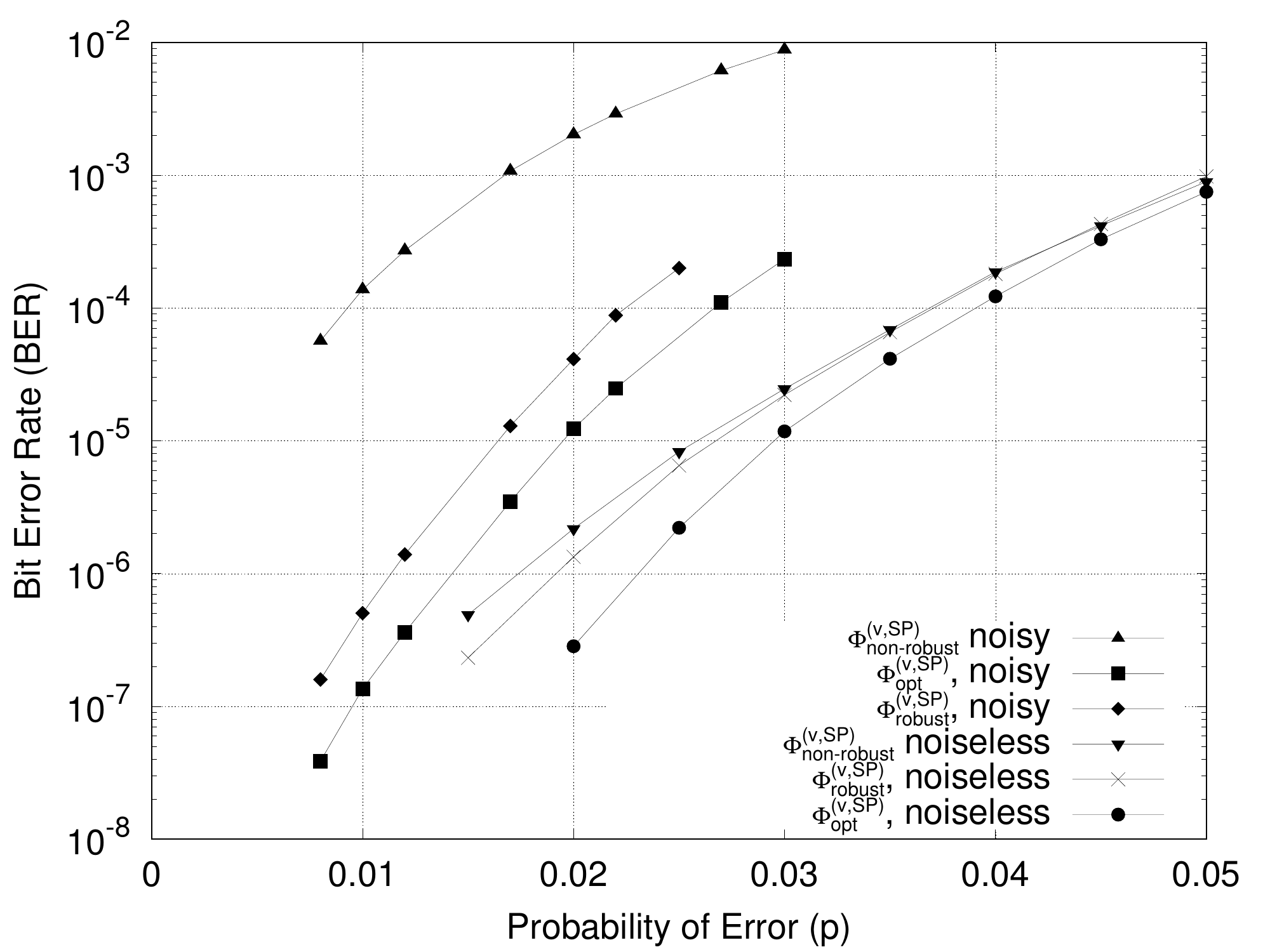}}
  \subfigure[~]{ \includegraphics[width=.48\linewidth]{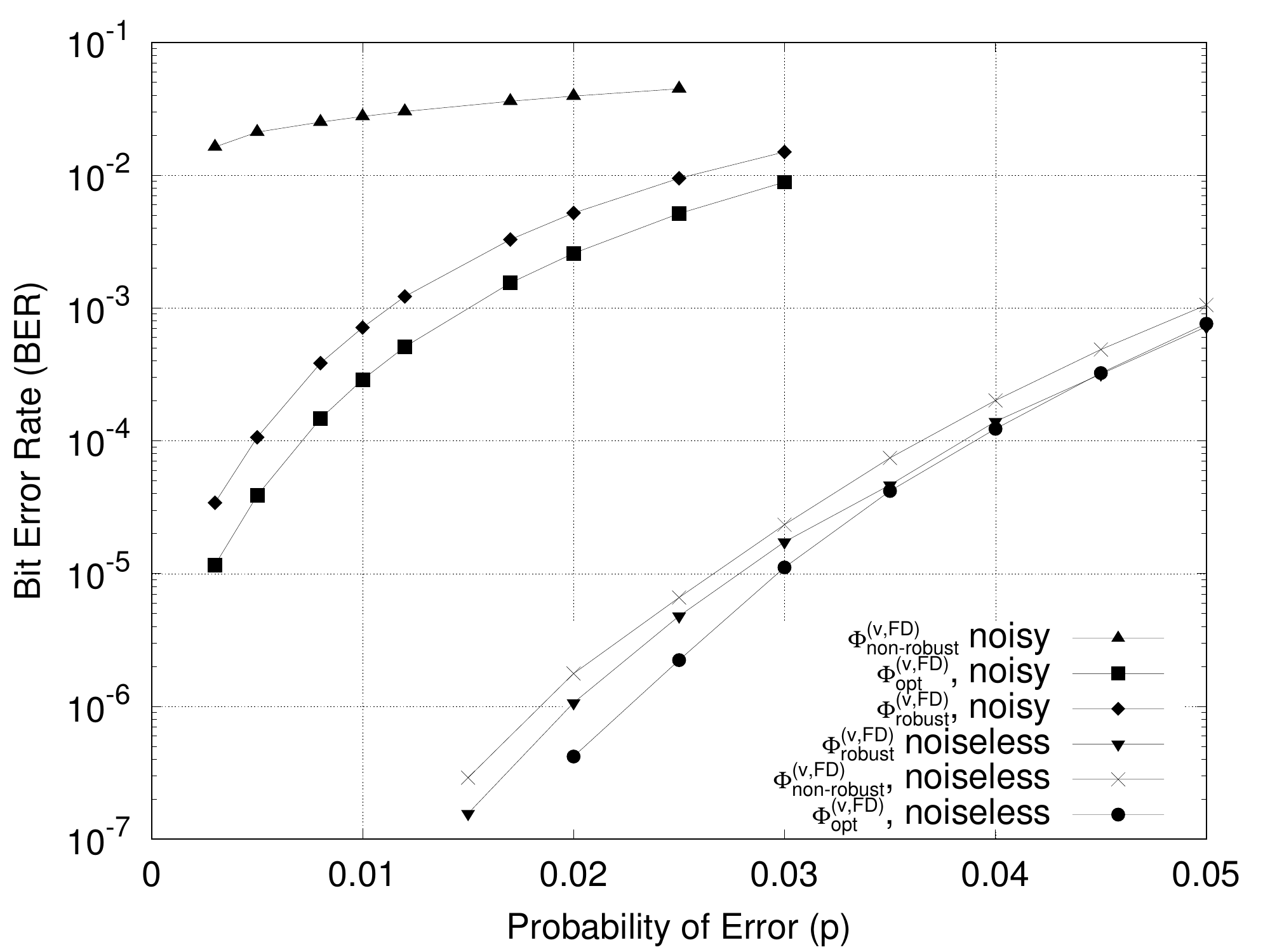}}
\end{center}
\caption{ (155, 93) Tanner Code, $d_v=3$, $d_c=5$, $100$ iterations, ~ (a) BER for the SP-Model, with $p_v = p_c = p_a = 0.05$, ~ (b)~BER for the FD-Model, with $p_v = p_c = p_a = 0.02$}
\label{fig:flnoisy}
\end{figure}

\begin{figure}
 \begin{center}
   \includegraphics[width=.48\linewidth]{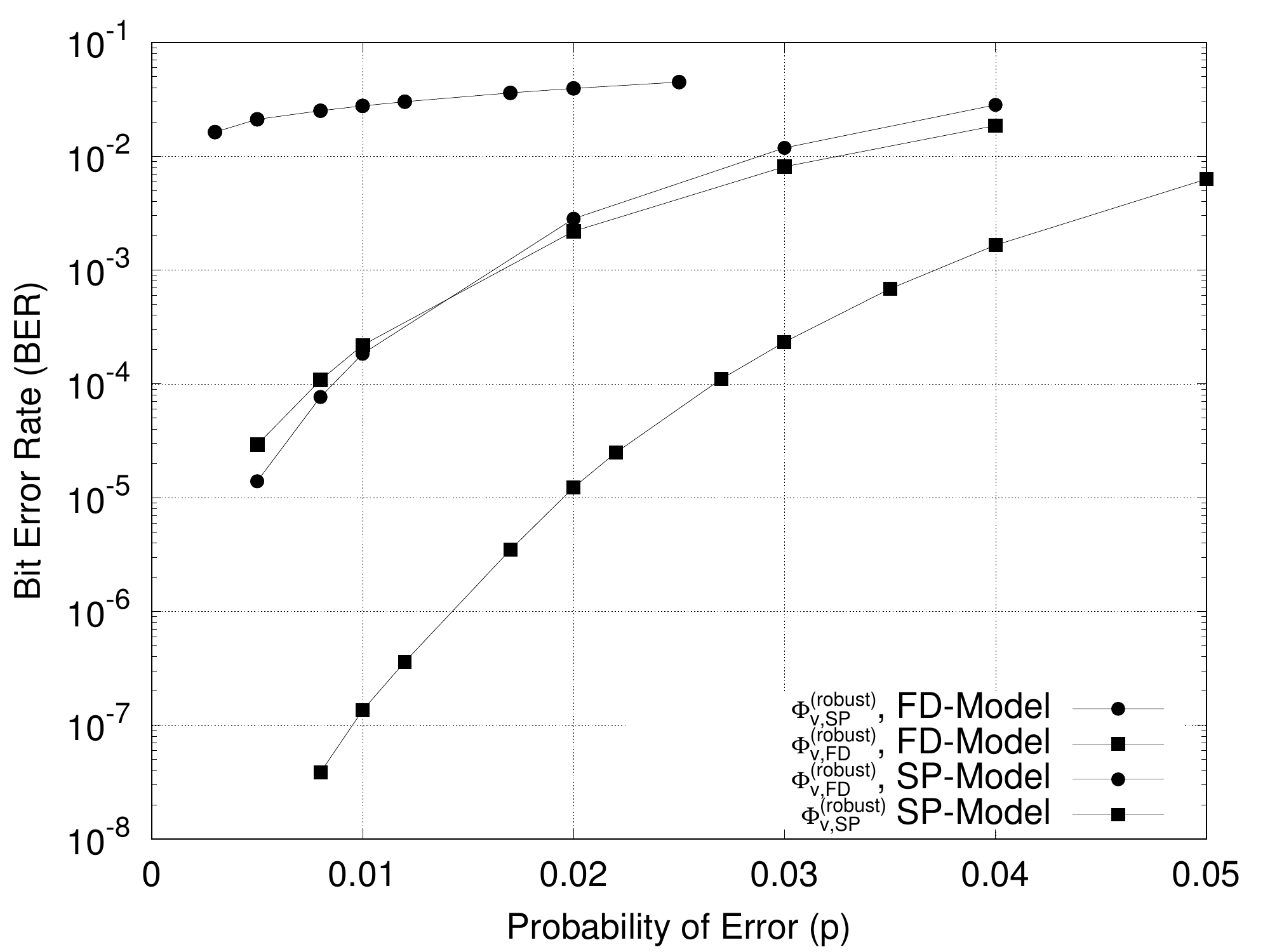}
 \end{center}
\caption{ (155, 93) Tanner Code, $d_v=3$, $d_c=5$, $100$ iterations, ~ $p_v = p_c = p_a = 0.05$ (SP-Model) and $p_v = p_c = p_a = 0.02$ (FD-Model) ~ For the legend, \emph{e.g.}, $\Phi_v^{(\text{robust,SP})}~ \text{FD} $ is the decoder robust for the SP-Model applied to the FD-Model }
\label{fig:compspfd}
\end{figure}

\section{Conclusion}\label{sec:conclusion}
In this paper, we performed an asymptotic performance analysis of noisy FAIDs using noisy-DE. 
We provided an analysis of the behavior of the functional threshold and showed that under restricted noise conditions, it enables to predict the asymptotic behavior of noisy FAIDs.
From this asymptotic analysis, we illustrated the existence of a wide variety of decoders robustness behaviors, and proposed a framework for the design of inherently robust decoders.
The finite-length simulations illustrated the gain in performance when considering robust decoders.


\appendix
The proof of Theorem~\ref{th:pe} follows the same steps as the proof of~\cite[Theorem 2]{li09IT}.
We first show that the symmetry is retained under faulty VN and CN processing.
We then show that the decoder error probability does not depend on the transmitted codeword.
For the sake of simplicity, the representation $0 \rightarrow 1$ and $1 \rightarrow -1$ is considered in the proof.
The all-zero codeword thus becomes the all-one codeword.

\subsection{Symmetry of the Faulty Iterative Processing}
Consider the two setups
\begin{enumerate}
 \item \emph{Setup 1:} The codeword $\mathbf{a} = [a_1,\dots,a_n]$, where $a_i \in \{-1,+1\}$, was transmitted, and the sequence $\mathbf{y} = [y_1,\dots, y_n]$ was received by the decoder.
 \item \emph{Setup 2:} The codeword $\mathbf{1} = [1,\dots, 1]$ was transmitted, and the sequence $\mathbf{a}. \mathbf{y} = [a_1 y_1, \dots, a_n y_n] $ was received.
\end{enumerate}
For Setup 1, denote $\mvtoc_{i,j}^{(\ell)}$ the message from a VN $i$ to a CN $j$ at iteration $\ell$ and denote $\mctov_{i,j}^{(\ell)}$ the message from a CN $j$ to a VN $i$ at iteration $\ell$.
Also denote $\gamma_{i}^{(\ell)}$ the APP message computed at node $i$ at iteration $\ell$.
We want to show that at any iteration $\ell$,
\begin{equation}	
P(\gamma_{i}^{(\ell)}|\mathbf{x}=\mathbf{a},\mathbf{y}) = P(a_i \gamma_{i}^{(\ell)}|\mathbf{x}=\mathbf{1},\mathbf{a}.\mathbf{y}) .
\end{equation}
The proof is made by recursion on the $\mvtoc_{i,j}^{(\ell)}$ and the $\mctov_{i,j}^{(\ell)}$.


\subsubsection{Initial messages}
The initial messages from VN $i$ to CN $j$ all verify
\begin{equation}
 P(\mvtoc_{i,j}^{(0)}|x_i=a_i,y_i)  = P(a_i \mvtoc_{i,j}^{(0)}|x_i=1,a_i y_i)
\end{equation}
by the channel symmetry~\cite[Definition 2]{varshney2011performance}.

\subsubsection{Check Node processing}
Assume that at iteration $\ell$, the condition
\begin{equation}\label{eq:condrec}
 P(\mvtoc_{i,j}^{(\ell)}|\mathbf{x}=\mathbf{a},\mathbf{y})  = P(a_i \mvtoc_{i,j}^{(\ell)}|\mathbf{x}=\mathbf{1},\mathbf{a}.\mathbf{y}) .
\end{equation}
is verified.
Then at any CN $j$,
\begin{equation}
 P(\mctov_{i,j}^{(\ell)}|\mathbf{x}=\mathbf{a},\mathbf{y}) = \sum_{\boldsymbol{\mvtoc}_j^{(\ell)}} \mathbf{P}^{(c)}(\mctov_{i,j}^{(\ell)}|\boldsymbol{\mvtoc}_j^{(\ell)}) \prod_{k=1}^{d_c-1} P(\mvtoc_{k,j}^{(\ell)}|\mathbf{x}=\mathbf{a},\mathbf{y})
\end{equation}
where $\boldsymbol{\mvtoc}_j^{(\ell)} = [\mvtoc_{1,j}^{(\ell)}, \dots, \mvtoc_{d_c-1,j}^{(\ell)}]$ is the set of VN messages incoming to the CN $j$.
The equality holds because the $\mvtoc_{k,j}^{(\ell)}$ are independent random variables.
Then, 
\begin{equation}
 P(\mctov_{i,j}^{(\ell)}|\mathbf{x}=\mathbf{a},\mathbf{y}) = \sum_{\boldsymbol{\mvtoc}_j^{(\ell)}} \mathbf{P}^{(c)}(a_j \mctov_{i,j}^{(\ell)}|\mathbf{a} \boldsymbol{\mvtoc}_j^{(\ell)})  \prod_{k=1}^{d_c-1} P(a_k \mvtoc_{k,j}^{(\ell)}|\mathbf{x}=1,\mathbf{a}.\mathbf{y})
\end{equation}
from~\eqref{eq:randcnusym},~\eqref{eq:condrec}, and $\prod_{k=1}^{d_c-1} a_k = a_j$.
By the variable change $\mvtoc_{k,j}'^{(\ell)} = a_k \mvtoc_{k,j}^{(\ell)}$, we finally get
\begin{align}\notag
 P(\mctov_{i,j}^{(\ell)}|\mathbf{x}=\mathbf{a},\mathbf{y}) &  = \sum_{\boldsymbol{\mvtoc'}_j^{(\ell)}} \mathbf{P}^{(c)}(a_j \mctov_{i,j}^{(\ell)}|\boldsymbol{\mvtoc'}_j^{(\ell)})  \prod_{k=1}^{d_c-1} P( \mvtoc_{k,j}'^{(\ell)}|\mathbf{x}=1,\mathbf{a} \mathbf{y}) \\ \label{eq:recCN}
							   & = P(a_j \mctov_{i,j}^{(\ell)}|\mathbf{x}=\mathbf{1},\mathbf{a}.\mathbf{y}) .
\end{align}

\subsubsection{Variable Node processing}\label{sec:VNproc}
At any VN $i$,
\begin{equation}
 P(\mvtoc_{i,j}^{(\ell+1)}|\mathbf{x}=\mathbf{a},\mathbf{y}) = \sum_{\boldsymbol{\mctov}_i^{(\ell)}} \mathbf{P}^{(v)}(\mvtoc_{i,j}^{(\ell+1)}|\boldsymbol{\mctov}_i^{(\ell)}) \prod_{j=1}^{d_v-1} P(\mctov_{i,j}^{(\ell)}|\mathbf{x}=\mathbf{a},\mathbf{y})
\end{equation}
where  $\boldsymbol{\mctov}_i^{(\ell)} = [\mctov_{i,1}^{(\ell)}, \dots, \mctov_{i,d_v-1}^{(\ell)}]$ is the set of CN messages incoming to the VN $i$.
Then
\begin{equation}
 P(\mvtoc_{i,j}^{(\ell+1)}|\mathbf{x}=\mathbf{a},\mathbf{y}) =  \sum_{\boldsymbol{\mctov}_i^{(\ell)}} \mathbf{P}^{(v)}(a_i \mvtoc_{i,j}^{(\ell+1)}|a_i \boldsymbol{\mctov}_i^{(\ell)}) \prod_{j=1}^{d_v-1} P(a_i \mctov_{i,j}^{(\ell)}|\mathbf{x}=\mathbf{1},\mathbf{a}.\mathbf{y})
\end{equation}
from~\eqref{eq:randvnusym} and~\eqref{eq:recCN}.
By the variable change $\mctov_{i,j}'^{(\ell)} = a_i \mctov_{i,j}^{(\ell)}$, we get
\begin{align}\notag
 P(\mvtoc_{i,j}^{(\ell+1)}|\mathbf{x}=\mathbf{a},\mathbf{y}) & =  \sum_{\boldsymbol{\mctov'}_i^{(\ell)}} \mathbf{P}^{(v)}(a_i \mvtoc_{i,j}^{(\ell+1)}|\boldsymbol{\mctov'}_i^{(\ell)}) \prod_{j=1}^{d_v-1} P({\mctov'}_{i,j}^{(\ell)}|\mathbf{x}=\mathbf{1},\mathbf{a}.\mathbf{y}) \\ \label{eq:recVN}
 & =  P(a_i \mvtoc_{i,j}^{(\ell+1)}|\mathbf{x}=\mathbf{1},\mathbf{a}.\mathbf{y})
\end{align}
which shows the recursion of~\eqref{eq:condrec}.

\subsubsection{APP processing}
At any VN $i$,
\begin{equation}\label{eq:recAPP}
 P(\gamma_{i}^{(\ell)}|\mathbf{x}=\mathbf{a},\mathbf{y}) = P(a_i \gamma_{i}^{(\ell)}|\mathbf{x}=\mathbf{1},\mathbf{a}.\mathbf{y}) .
\end{equation}
The proof is obtained from the previous recursion on VN and CN processing, and following the steps of VN processing.

\subsection{Error Probability}
We now show that the error probabilities of Setup 1 and Setup 2 are equal.

\subsubsection{Error probability at node $i$}
For Setup 1, the error probability at VN $i$ conditionally to $\mathbf{y}$ is
\begin{equation}
 P_{e,i}^{(\ell)}(\mathbf{x} = \mathbf{a}, \mathbf{y}) = \int_{\Omega_i} P(\gamma_{i}^{(\ell)}|\mathbf{x}=\mathbf{a},\mathbf{y}) d\gamma_{i}^{(\ell)}
\end{equation}
where $\Omega_i = \mathbb{R}^{-}$ if $a_i=1$, and $\Omega_i = \mathbb{R}^{+}$ if $a_i=-1$.
Then, from~\eqref{eq:recAPP},
\begin{equation}
 P_{e,i}^{(\ell)}(\mathbf{x} = \mathbf{a}, \mathbf{y}) = \int_{\Omega_i} P(a_i \gamma_{i}^{(\ell)}|\mathbf{x}=\mathbf{1},\mathbf{a}.\mathbf{y}) d\gamma_{i}^{(\ell)} .
\end{equation}
By variable change $ {\gamma'}_{i}^{(\ell)} = a_i \gamma_{i}^{(\ell)} $, we get
\begin{equation}\label{eq:indPe} 
  P_{e,i}^{(\ell)}(\mathbf{x} =\mathbf{a}, \mathbf{y}) = \int_{\mathbb{R}^{-}} P({\gamma'}_{i}^{(\ell)}|\mathbf{x}=\mathbf{1},\mathbf{a}. \mathbf{y}) d{\gamma'}_{i}^{(\ell)} 
  =  P_{e,i}^{(\ell)}(\mathbf{1},\mathbf{x} = \mathbf{a}.\mathbf{y}) .
\end{equation}

\subsubsection{Error probability}
The error probability of Setup $1$ is given by
\begin{equation}
 P_{e}^{(\ell)}(\mathbf{a}) = E_{i,\mathbf{y}}\left[ P_{e,i}^{(\ell)}(\mathbf{x} =\mathbf{a}, \mathbf{y}) \right] = E_{i,\mathbf{y}}\left[ P_{e,i}^{(\ell)}(\mathbf{x} =\mathbf{1}, \mathbf{a}.\mathbf{y}) \right] .
\end{equation}
By the variable change $\mathbf{y}' =\mathbf{a}\mathbf{y}$, we get
\begin{equation}
 P_{e}^{(\ell)}(\mathbf{a}) = E_{i,\mathbf{y'}}\left[ P_{e,i}^{(\ell)}(\mathbf{x} =\mathbf{1},\mathbf{y'}) \right]
\end{equation}
and
\begin{equation}
  P_{e}^{(\ell)}(\mathbf{a}) =  P_{e}^{(\ell)}(\mathbf{1})
\end{equation}
which concludes the proof.


\bibliographystyle{IEEEtran}
\bibliography{biblio}
\nopagebreak

\end{document}